\documentclass[11pt]{article}
\usepackage{amsmath, amsthm, amssymb}
\usepackage{times}
\usepackage{fullpage}
\usepackage{graphics}
\usepackage{graphicx}
\usepackage{units}
\usepackage[tight]{subfigure}
\usepackage{graphicx}
\usepackage{xspace}
\usepackage{cite}

 \advance \topmargin by -\headheight
 \advance \topmargin by -\headsep
 
 \evensidemargin \oddsidemargin
\def\eps{{\epsilon}}

\def\Prf{{\rm Pr}}

\newcommand{\xhdr}[1]{\paragraph{{#1}.}}

\newcommand{\Exp}[1]{
E\left[{#1}\right]
}

\newcommand{\Prgt}[2]{
\Prf\left[\mbox{\em #1}~|~\mbox{\em #2}\right]
}

\newcommand{\omt}[1]{}
\newcommand{\Xomit}[1]{}
\newcommand{\supproof}[1]{}
\newcommand{\supproofeq}[1]{}

\newlength{\saveparindent}
\newlength{\saveparskip}
{\end{list}}

\newtheorem{fact}{Fact}[section]
\newtheorem{lemma}[fact]{Lemma}
\newtheorem{theorem}[fact]{Theorem}
\newtheorem{definition}[fact]{Definition}

\newtheorem{observation}[fact]{Observation}
\newtheorem{conjecture}[fact]{Conjecture}

\newcommand{\comment}[1]{}
\newcommand{\say}[1]{\begin{center}\fbox{\parbox{\textwidth}{#1}}\end{center}}

\DeclareMathOperator{\var}{Var}

\def\qstar{q^*}

\begin{document}

\def\df{\dfrac}
\def\ofrac{\dfrac}
\newcommand{\fn}[1]{
{{#1}(\cdot)}
}
\def\ds{\displaystyle}
\def\os{\ds}

\title{
Voting with Limited Information and Many Alternatives
}

\author{
Flavio Chierichetti
\thanks{
Department of Computer Science,
Cornell University, Ithaca NY 14853.
Supported in part by 
NSF grant CCF-0910940. 
}
\and
Jon Kleinberg
\thanks{
Department of Computer Science,
Cornell University, Ithaca NY 14853.
Supported in part by 
a John D. and Catherine T. MacArthur Foundation Fellowship,
a Google Research Grant,
a Yahoo!~Research Alliance Grant,
and NSF grants 
IIS-0705774, 
IIS-0910664, 
and CCF-0910940. 
}
}

\date{October, 2011}

\begin{titlepage}
\maketitle

\begin{abstract}
The traditional axiomatic approach to voting is motivated by the problem of
reconciling differences in subjective preferences.
In contrast, a dominant line of work in 
the theory of voting over the past 15 years
has considered a different kind of scenario, also fundamental to voting,
in which there is a genuinely ``best''
outcome that voters would agree on if they only had enough information.
This type of scenario has its roots in the classical Condorcet Jury Theorem;
it includes cases such as jurors in a criminal trial who all want 
to reach the correct verdict but disagree in their inferences 
from the available evidence,
or a corporate board of directors who all want to improve the
company's revenue, but who have different information that
favors different options.

This style of voting leads to a natural set of questions:
each voter has a {\em private signal} that provides probabilistic
information about which option is best, and a central question is
whether a simple plurality 
voting system, which tabulates votes for different options,
can cause the group decision to arrive at the correct option.
We show that plurality voting is powerful enough to achieve this:
there is a way for voters to map their signals into votes for options
in such a way that --- with sufficiently many voters ---
the correct option receives the greatest number of votes with high probability.
We show further, however, that any process for achieving this is 
inherently expensive in the number of voters it requires:
succeeding in identifying the correct option with probability
at least $1 - \eta$ requires $\Omega(n^3 \eps^{-2} \log \eta^{-1})$
voters, where $n$ is the number of options and $\eps$ is a 
distributional measure of the minimum difference between the options.

\end{abstract}

\thispagestyle{empty}
\end{titlepage}

\section{Introduction}
\label{sec:intro}

\xhdr{Information-Based Voting}
A dominant recent theme in the study of voting has been to trace
differences in voters' preferences back to differences in 
the information they have about the world.
This information-based approach has its roots
in one of the earliest results in voting theory --- the
{\em Condorcet Jury Theorem}, which used the then-young theory of probability
to model a situation in which a panel of jurors each wants
to vote for the correct decision in a trial, but each juror may be wrong 
about what the correct decision is independently 
and with probability $p < \frac12$ \cite{young-condorcet}.
It is only very recently, however, that this approach has received 
deeper theoretical attention
\cite{austen-smith-jury-voting,bhattacharya-strat-voting,
feddersen-swing-voter,feddersen-jury,
myerson-jury-voting},
leading to what is now a large and growing body of research.

The basic premise of the information-based approach to voting is that
all voters want the best option for the group as a whole, but they
disagree on what this best option is, based on the information they have.
This models a wide range of situations where the differences among voters
are not purely subjective, but instead based on uncertainty.
For example, in most criminal trials the key question is genuinely
whether the defendant committed the crime or not; all jurors want
to reach the correct decision, but they disagree on which of the 
pieces of information presented at the trial are most salient.
Similarly, all the members of a corporate board of directors may
genuinely agree that the goal is to reach a decision that will most
improve the company's future revenue, but they disagree
on which course of action is most likely to achieve this.
Even at the level of large populations, there can be cases where each
voter wants a candidate whose election --- for example ---
will lead to the strongest improvement in the economy, but there is
disagreement among the voters about which candidate is most likely 
to achieve this.

This view of voters as information-processing agents trying
to reach a correct decision has made it possible to develop 
models for a range of important phenomena in voting; 
these include the fact that voters realize they might be wrong,
and the corollary that they can sometimes be convinced by evidence
\cite{austen-smith-jury-voting,feddersen-jury},
the corresponding role of deliberation in committee voting
\cite{gerardi-deliberation},
and the fact that many voters may choose to abstain or not participate
when they believe that others have more accurate information than they do
\cite{battaglini-swing-voter-lab,feddersen-abstention}.

\xhdr{A Basic Model of Information-Based Voting}
In this paper, we consider the following basic theoretical model that
has received wide study
\cite{austen-smith-jury-voting,feddersen-jury}.
There is a decision to be made, involving selecting from among
several possible {\em options} $A_1, \ldots, A_n$.
One of these options is {\em correct}, and all voters want to select it.
However, which option is correct is determined by a
process that cannot be directly observed, and
the voters have to use indirect signals to infer the correct option.
Before casting a vote, each voter $t$ receives a private
{\em signal} equal to some value $s_j$, providing evidence
for the identity of the correct option.
(The full set of possible signals will be labeled
$\{s_1, s_2, \ldots, s_C\}$.)
We assume that certain kinds of signals are more plentiful when
certain options are correct, and that voters know conditional
probabilities of the form 
$\Prgt{$s_j$ is received}{$A_i$ is correct} = \rho_{ij}.$
We further assume that no two options induce exactly the same
set of conditional probabilities over signals.
Based on the signal she receives, each voter casts a vote for one option,
potentially using a randomized rule to map the signal to a vote.
A {\em voting system} --- a rule for mapping a collection of votes to
a group decision --- is then applied to these votes.
We are interested in the probability that the group decision
will be equal to the correct option $A_i$.

Much of the power of this model in economics and political science
comes from the way in which it separates the {\em signals} received by
the voters from the {\em options} they are voting on. This captures a
basic property of voting in many real-life situation, including the
ones described at the outset: the signals represent information
and decision-making heuristics that the individual voters possess in
their minds, while the options correspond to candidates or alternatives
presented on a ballot.  For many reasons, the institution of voting therefore
does not (and generally cannot) consist of a simple sharing of everyone's
signals.  Instead, voters are only able to convey the information they
possess in a more indirect fashion, by voting for one of the given options.
The crucial question is whether there is a (possibly
randomized) algorithm each voter can apply to his or her signal to
produce a vote, in such a way that the correct option is chosen.

For simplicity in the following discussion, we consider
an equivalent formulation of this model 
(in the spirit of \cite{anderson-herding-expmt-aer}),
via an experiment involving urns and marbles.
Suppose an experimenter has a collection of urns $A_1, \ldots, A_n$,
and each urn contains marbles of colors $s_1, \ldots, s_C$.
The fraction of marbles of color $s_j$ in urn $A_i$ is equal to $\rho_{ij}$;
no two urns have exactly the same mixture of colors.
Now, the experimenter announces to a set of test subjects that he is placing
one of the urns $A_1, \ldots, A_n$ on a table.
Each test subject draws and replaces a single marble from
the urn on the table, without showing it to the other subjects,
and then writes down a vote (on a secret ballot) 
for which urn she believes is on the table.
The experimenter then applies a voting rule to these votes,
producing a group decision, and awards the group a prize
if this group decision is equal to the true urn that was on the table.
It should be clear that this formulation is simply a rephrasing of
the original model, with the urns representing the options and
the marbles representing the signals.

\xhdr{Can Plurality Voting Produce the Correct Answer?}
Much of the initial theoretical work on these issues focused on the
case of
$n = 2$ options --- that is, voting when there are two alternatives,
such as in a jury trial or a yes/no vote on a proposed rule
\cite{austen-smith-jury-voting,bhattacharya-strat-voting,
feddersen-swing-voter,feddersen-jury, myerson-jury-voting}.
But many settings involve more than two options, and in this case
the following basic question has remained open.
Suppose the votes will be aggregated simply using plurality voting,
with the urn receiving the most votes chosen as the group decision.
Is there a rule the voters can use for mapping signals (colors) to
votes,
such that for any instance of the problem with urns $A_1, \ldots,
A_n$,
a set of $m$ voters will identify the correct urn with
probability converging to $1$ as $m$ grows?
And if so, how large a set of voters is needed to guarantee
a success probability of $1 - \eta$, for a given $\eta > 0$?

Recent work has highlighted the challenge and general lack of
understanding of
this question with more than two options, raising it as an
open problem and providing interesting results in highly structured
special cases where the signal space is rich enough that
each option has a disjoint set of one or more
signals that uniquely favor it
\cite{hummel-jury-voting,hummel-multicandidate}.
For general sets of signals, the question has been open:
if the signals are expressively weak compared to the full set of
options,
is there necessarily any strategy for mapping signals to votes
that would lead to the correct outcome under a simple system
like plurality voting?

\xhdr{Optimal Information-Based Voting: Main Results}
Our first main result is that 
for any finite set of signals, and any finite set of
options that induce distinct distributions over these signals,
there is a strategy such that a sufficiently large set of voters
can arrive at the correct option with high probability using plurality voting.
In other words, each voter translates her signal into a vote in such
a way that the option receiving the most votes is, with high probability,
the correct one.  

Second, we show that achieving this goal using
plurality voting is very expensive: it 
requires a large number of voters.
We give lower and upper bounds on the number of voters needed to 
achieve a high probability of correctness, parametrized by three quantities:
the number of options, the number of signals, and a quantity
measuring the minimum separation between the 
distributions over signals induced by any two options.
The lower bound is the technically most involved of our results, and 
for two signals it is asymptotically tight in both the 
number of options and the separation parameter.

The technical core of our results is the case in which there are $n$
options and $2$ signals.  Let $\eps$ be the minimum positive
difference between the probability assigned to a fixed signal $s_k$ by
two different options $i$ and $j$.  With two signals, we show there is
a strategy by which $O(n^3 \eps^{-2} \log \eta^{-1})$ voters can
arrive at the correct option using plurality voting with probability
at least $1 - \eta$.  The strategy is symmetric, in that all voters
map signals to votes according to the same probabilistic rule.  While
the algorithm involves a carefully designed rule, it is based on a
principle that is intuitively natural: the voters ``hedge'' against
the possibility that their information points in the wrong direction,
by sometimes choosing to vote for an option other than the one
supported by their signal. The bound achieved by our algorithm is
tight: there are instances in which $\Omega(n^3 \eps^{-2} \log
\eta^{-1})$ voters are necessary to achieve such a guarantee; this
lower bound applies even to asymmetric strategies in which different
voters can use different rules.

Note that by the pigeonhole principle, the minimum difference $\eps$
is at most $1/(n-1)$, and hence $\eps^{-1}$ is a parameter that is at
least as large as $n-1$. For example, the special case with urns $A_0,
A_1, \ldots, A_n$, in which $A_i$ contains $i$ blue marbles and $n-i$
red marbles, has $\eps = 1/n$, and so for this problem the tight bound
on the minimum number of voters needed is $\Theta(n^5 \log
\eta^{-1})$.

A recurring theme in our results is this fifth-power dependence of the
number of voters on $n$, in the case when $\eps^{-1}$ is close to $n$.
As such, it is useful to provide some intuition at the outset for how
this fifth-power dependence arises. Thus, the following description is
deliberately informal, but gives a sense for where this functional
form comes from.  Let there be $m$ voters, and for simplicity let us
consider the special case from the previous paragraph, with urns $A_0,
A_1, \ldots, A_n$, in which $A_i$ contains $i$ blue marbles and $n-i$
red marbles. Under the asymptotically optimal (randomized) algorithms
we consider, the correct urn will receive a greater number of votes in
expectation than any other urn; this is why, with enough voters, we
will eventually be able to distinguish the correct urn using plurality
voting. Now, we will find that the optimal algorithm has the following
two properties. First, it spreads out the votes relatively uniformly
across a set of $\Theta(n)$ urns, and so if there are $m$ voters, each
of the urns in this set receives $\Theta(m/n)$ votes in expectation.
The second, subtler property is the crucial one: the optimal algorithm
ensures that the correct urn receives the most votes in expectation
using a delicate optimization under which the the expected number of
votes received by the correct urn will exceed the expected number of
votes received by the adjacent urns by a factor of only $1 + \delta$,
where $\delta = \Theta(n^{-2})$. As a result, to distinguish the
correct urn with high probability, we need a number of samples that is
sufficient to yield at least $\Theta(\delta^{-2}) = \Theta(n^4)$ votes
for the correct urn. But since the correct urn receives only
$\Theta(m/n)$ votes in expectation, this means that we need $m$ to be
$\Theta(n^5)$.

We observe that in our more general bound $O(n^3 \eps^{-2} \log
\eta^{-1})$, the form of the dependence on $\eps^{-1}$ is in fact
necessary even if the voters could share their signals (rather than
casting individual votes). Indeed, with $n = 2$ options that assign
probabilities to signals differing by only $\eps$, even a single
observer would need to see $\Theta(\eps^{-2} \log \eta^{-1})$ signals
in order to identify the correct option with probability at least $1 -
\eta$. Thus, with a constant number of options, plurality voting is
allowing voters to aggregate their information with an efficiency that
is within a constant factor of the efficiency achievable by a single
person who could observe all signals directly.

For the case of $C > 2$ possible signals or colors, let $\epsilon$ denote the
minimum $\ell_1$ distance\footnote{We observe that, in the multicolor case, choosing
  the right parameter to define a notion of ``distance'' between urns
  is not as straightforward as in the bichromatic case. We chose $\ell_1$
 because it is the parameter that has been used in the
  literature to determine the minimum number of samples that allows 
  an algorithm to distinguish between probability
  distributions.} of two distinct urns' probability
distributions. We have an upper bound of
$O\left( (C\log C)^2 n^3 \eps^{-2} \log \frac n{\eta}\right)$
on the number of voters needed.
Since the lower bound for the two-signal case applies with $C > 2$,
it is tight in $\eps$, and we lose only an exponentially small factor
in $n$. Finding the correct dependence of the required number of voters
on $n$ and $C$ is an interesting open question.

Under plurality voting, voters can only communicate the name of
a single option in response to a signal.  We also consider voting
systems that allow voters to be much more expressive:
{\em cumulative voting}, in which each vote consists of assigning a
non-negative weight to each option (such that the weights sum to $1$);
and {\em Condorcet voting}, in which each vote consists of a ranking of
all the options.
For bichromatic urns, we show that cumulative voting requires only $O(\eps^{-2} \log \eta^{-1})$
voters in order to succeed with high probability; this is tight even
compared to the baseline discussed above, when a single observer has
access to all the signals.
We show that a similar bound would hold for Condorcet voting,
modulo an intriguing conjecture about distributions over permutations.

\xhdr{Optimal Information-Based Voting: Main Techniques}
The possibility result for identifying the correct option 
is based on a technique that implicitly draws
a connection to the framework of {\em proper scoring rules} from statistics
\cite{gneiting-scoring-rules}.
Proper scoring rules can be thought of as incentive systems for eliciting
accurate probabilistic forecasts from expert predictors; 
the contexts in which they have been
used in earlier work are quite different from ours, and to our knowledge
there have not been previous linkages between proper scoring rules
and information-based voting.

A construction based on proper scoring rules provides the first
method for obtaining the correct option using plurality voting.
However, we need to go beyond this construction in order to obtain
a tight bound on the number of voters needed:
in a sense to be made precise below, we can prove that 
any direct use of proper scoring rules in our setting requires at
least $\Omega(n \eps^{-4})$ voters to achieve a high probability of success.
This is at least $\Omega(n^3 \eps^{-2})$ since $\eps \leq (n-1)^{-1}$, and 
more significantly, it has an asymptotically sub-optimal dependence on
$\eps$ of $\Omega(\eps^{-4})$ when $n$ is a constant, 
whereas our stronger approach achieves the
optimal dependence of $\Theta(\eps^{-2} \log \eta^{-1})$ for constant $n$.

For the lower bound, we need to show that with $O(n^3 \eps^{-2} \log
\eta^{-1})$ voters, there is a probability $\eta$ that plurality
voting will choose the wrong option.  For this, we identify a natural
``close competitor'' $j$ of the correct option $i$, with a very
similar signal distribution, and we consider a random variable that
measures the extent to which the number of votes for the correct
option $i$ exceed the number for this competitor $j$.  The (possibly
asymmetric) strategies of the voters determine the variance of this
random variable, and roughly speaking we follow a two-pronged argument
in terms of this variance.  If this variance is too low, then there is
a high chance that voters would behave the same regardless of whether
the option generating the signals was $i$ or $j$, and hence that if
they are correct about $i$ with high probability, then they would have
to be wrong with constant probability when $j$ is the correct option.
If the variance is above a certain low threshold, on the other hand,
then we apply a carefully tuned ``anti-concentration'' inequality from
\cite{f43,mv} showing that there is a constant probability that the
number of votes for $i$ will drop below the number for its competitor $j$.

\xhdr{Further Related Work}
Finally, we mention two other recent lines of work that have
also considered the problem faced by a set of agents trying to agree
on a joint decision from a set of alternatives.
Mossel, Sly, and Tamuz study a version of the problem in which
there are two options, and each agent is given a probabilistic signal 
providing information about which option is correct
\cite{mossel-agreement}; in contrast
to our approach and to the work on voting discussed above,
they consider a model in which agents may communicate iteratively
over multiple rounds.
Caragiannis and Procaccia consider a setting based on agents
that possess utilities over options;
within this framework, they show that simple voting rules 
can approximately optimize the sum of agents' utilities for
the option that is selected
\cite{caragiannis-voting}.

\section{An Upper Bound with Two Signals}\label{sec:twosig}

We begin by considering the case of two signals.
Suppose we have a collection of $n$ urns, labeled $p_1, \ldots, p_n$, the $i$-th of which having a $p_i$ fraction of blue balls and a $1-p_i$ fraction of red balls, with $p_1 \le p_2 \le \cdots \le p_n$.
We let 
$\eps$ denote the smallest difference between two consecutive $p_i$'s:
$\eps = \min_{0 \le i \le n - 1} \left(p_{i+1} - p_i\right).$

\medskip

We assume that one urn is adversarially chosen as the correct one
(we will also refer to this as the {\em unknown} urn).
Then each player
draws a ball from the urn and votes for the name of an urn 
based on the color they observe.

\medskip

We describe the strategy that the players will use to randomly
choose which vote to cast:
\begin{itemize}
\item[1.] Let 
$\ds{b_k = \sum_{\ell=1}^{k-1} \frac{2-(p_{\ell+1} + p_{\ell})}{p_{\ell+1} - p_{\ell}}}$ and
$\ds{r_k = \sum_{\ell=k}^{n-1} \frac{p_{\ell+1} + p_{\ell}}{p_{\ell+1} - p_{\ell}}}.$
Then define $R = \sum_{k=1}^n  r_k$ and $B = \sum_{k=1}^n
  b_k$, and set $M = \max(R, B)$.
\item[2.] The probability that a voter will vote for $p_j$ if a red ball is drawn is
$R_j = M^{-1} \cdot \left(r_j + \frac{M -R}{n}\right).$
\item[3.] The probability that a voter will vote for $p_j$ if a blue ball is drawn is
$B_j = M^{-1} \cdot \left(b_j + \frac{M -B}{n}\right).$
\end{itemize}

It is easy to check that the two given distribution are indeed probability distributions (their values are non-negative and they both sum up to one).
Now, the
probability that a player will vote for $p_j$  given that the correct,
adversarially chosen, distribution is $p_i$, is
$$\Pr_{\substack{\mathbf{X} \sim (p_i,1-p_i)\\ \mathbf{P} \sim f(\mathbf{X})}}\left[\mathbf{P} =
    p_j\right]  =  p_i \Pr_{\mathbf{P} \sim
  f(\text{blue})}[\mathbf{P} =  p_j] + (1-p_i) 
\Pr_{\mathbf{P} \sim f(\text{red})}[\mathbf{P} = p_j] = p_i B_j + (1-p_i) R_j = E_i(j).$$

Now consider two urns, $p_i$ and $p_j$. We compute
the difference between the probabilities that a vote for urn $p_{i}$
and a vote for urn $p_{j}$ are cast, given that the correct urn is $p_i$:
$$\Delta_{i}(j) = E_i(i) - E_i(j) =
p_i \left(B_i
-B_j\right)
 + (1-p_i)  \left(R_i - R_j\right).$$
We will lower-bound $\Delta_{i}(j)$ to bound the number of voters needed
to let the voting scheme be successful with high probability. Suppose
first that $j < i$; then
$$
M \cdot \Delta_{i}(j) = p_i \cdot \sum_{\ell=j}^{i-1}
\frac{2-(p_{\ell+1} + p_{\ell})}{p_{\ell+1}-p_{\ell}} - (1-p_i) \cdot
\sum_{\ell=j}^{i-1} \frac{p_{\ell+1}+p_{\ell}}{p_{\ell+1}-p_{\ell}} 
= \sum_{\ell=j}^{i-1} \frac{2 \cdot p_i - (p_{\ell+1}+p_{\ell})}{p_{\ell+1}-p_{\ell}},
$$
observing that in each term of the sum we have $p_i \ge
p_{\ell+1}$, since $\ell \le i - 1$. Therefore,
$$
M \cdot \Delta_{i}(j)  \ge \sum_{\ell=j}^{i-1} \frac{2 p_{\ell+1}
  -(p_{\ell+1} + p_{\ell})}{p_{\ell+1}-p_{\ell}}
= \sum_{\ell=j}^{i-1} \frac{p_{\ell+1} -
  p_{\ell}}{p_{\ell+1}-p_{\ell}}= i - j.
$$
If, instead, $i < j$ we have:
\begin{align*}
M \cdot \Delta_{i}(j) & = - p_i \cdot \sum_{\ell=i}^{j-1}
\frac{2-(p_{\ell+1} + p_{\ell})}{p_{\ell+1}-p_{\ell}} + (1-p_i) \cdot
\sum_{\ell=i}^{j-1} \frac{p_{\ell+1}+p_{\ell}}{p_{\ell+1}-p_{\ell}} \\
& = \sum_{\ell=i}^{j-1} \frac{(p_{\ell+1}+p_{\ell}) -2p_i}{p_{\ell+1}-p_{\ell}}
\ge \sum_{\ell=i}^{j-1} \frac{(p_{\ell+1}+p_{\ell}) -2p_{\ell}}{p_{\ell+1}-p_{\ell}}
= \sum_{\ell=i}^{j-1} \frac{p_{\ell+1}
  -p_{\ell}}{p_{\ell+1}-p_{\ell}} = j - i,
\end{align*}
where the inequality follows from $p_i \le p_{\ell}$.
Therefore for $j \ne i$, we have
\begin{equation}\label{deltaij}
\Delta_{i}(j) \ge \frac{\left|i - j\right|}M.
\end{equation}

\medskip

We now give an upper bound on the probability that the correct urn
will be chosen by a voter. 
Note, somewhat counter-intuitively, that the probability of a correct vote
is higher when this upper bound is smaller --- this is because the
$\Delta_{i}(j)$ are additive gaps, not
multiplicative ones, and so by making the upper bound on the expected
number of votes for the correct urn smaller, the gap $\Delta_{i}(j)$
becomes larger relative to the mean.

\medskip

Recall that the correct urn is $p_i$. We upper-bound the
probability that a vote will go to $p_i$:
\begin{align*}
E_i(i)  & = p_i \cdot M^{-1} \cdot \left(b_i + \frac{M -B}{n}\right) + (1 -
p_i) \cdot M^{-1} \cdot \left(r_i + \frac{M -R}{n}\right)\\
& \le p_i \cdot \left(M^{-1} \cdot b_i + \frac{1}{n}\right) + (1 -
p_i) \cdot \left(M^{-1} \cdot r_i + \frac{1}{n}\right).
\end{align*}
Observe that, by the definition of $\epsilon$, we have that
$\ds{b_i = \sum_{\ell=1}^{i-1} \frac{2-(p_{\ell+1} +
  p_{\ell})}{p_{\ell+1} - p_{\ell}}}$ satisfies $b_i \le
\frac{2i}{\epsilon}$, and furthermore that $\ds{r_i = \sum_{\ell=i}^{n-1}
\frac{p_{\ell+1} + p_{\ell}}{p_{\ell+1} - p_{\ell}} \le
\frac{2(n-i)}{\epsilon}}$. Thus,
$$
E_i(i)   \le p_i \cdot \left(\frac{2i}{\epsilon M} + \frac{1}{n}\right) + (1 -
p_i) \cdot \left(\frac{2 (n-i)}{\epsilon M} + \frac{1}{n}\right)
\le \frac{2n}{\epsilon M} + \frac1{n}.
$$

We now give an upper bound on $M$. This will allow
us to apply a Chernoff bound and finish the proof.

\medskip

Recall that $M = \max(R,B)$; we will upper bound $R + B$ to get an
upper bound on $M$:
$$ R + B = \sum_{k=1}^{n} (r_k + b_k) \le \sum_{k=1}^n (r_1 + b_n) =
n \cdot (r_1 + b_n) = n \cdot
\sum_{\ell=1}^{n-1} \frac{2}{p_{\ell+1}-p_{\ell}} \le 2 \cdot \frac{n(n-1)}{\epsilon}.$$

It follows that
\begin{equation}\label{Mbound}
M \le \frac{2n(n-1)}{\epsilon}.
\end{equation}

\comment{
$p_{\left\lceil\frac n2\right\rceil} \le \nicefrac12$
($p_{\left\lfloor\frac n2\right\rfloor} \ge \nicefrac12$). Then, at
least
$\nicefrac n8$ of those $p_i$'s will be such that $p_{i+1} -
p_{i} \le \frac2n$ --- for otherwise they would span more than just the
interval $\left[0,\nicefrac12\right]$. Then, $B = \sum_{k=0}^n (N_k
\cdot b_k) \ge $}

Therefore, going back to the probability that an urn identical to 
the correct urn is voted for, we have
$$
E_i(i)  \le  \frac{2n}{\epsilon M} + \frac1{n}
 \le \frac{2n}{\epsilon M} + \frac1{n} \cdot \frac{2n(n-1)}{M
  \epsilon}
 \le \frac {4n}{\epsilon M}.
$$
Furthermore, since $\Delta_{i}(j) > 0$ for each $j \ne i$, we have that
the urn $p_i$ is the most likely urn to be voted for, and therefore
$E_i(i) \ge \frac{1}{n}.$

We are now ready to state the main theorem of the section. Its proof
employs a careful application of the Chernoff bound, and the
inequalities we have derived in this section.
\begin{theorem}\label{thm:upperb_pv}
Let urns $p_1, p_2, \ldots, p_n$ be given, with urn $p_i$ having a
$p_i$ fraction of blue balls, and a $1-p_i$ fraction of red balls. Let $0 \le p_1 < p_2 < \cdots < p_n \le 1$.
Also, let $\epsilon$ be $\epsilon = \min_{1 \le i \le n-1} (p_{i+1} -
p_i)$. Then, for Plurality Voting, $O\left(n^3\epsilon^{-2} \ln \eta^{-1}\right)$ voters are
sufficient to guarantee a probability of at least $1 -\eta$ that
the correct urn receives the most votes.
\end{theorem}
\begin{proof}
Observe that $\epsilon \le \frac1{n-1}$.
Choose some $\eta \in (0,1)$, and suppose the number of players is
$ m = \left\lceil 108 \cdot \frac{M (n-1)}{\epsilon} \cdot \ln \frac4{\eta}\right\rceil$
 --- we will show that $m$
players will be enough to 
choose the correct option
with probability is at least $1 - \eta$. Observe that, given
our upper bound $M \le \frac{2n(n-1)}{\epsilon}$, we have
$$m \le  \left\lceil 216 \cdot \frac{(n-1)^2n}{\epsilon^2}\right\rceil.$$

Recall that we say that the players lose if an urn with a different
distribution from the unknown urn wins the election. 
We will upper-bound the probability that the players
lose, using the following form of the Chernoff bound:

\begin{theorem}[Chernoff bound]
Let $X_1, \ldots, X_m$ be independent $0/1$ random variables with
expectation $E[X_i] = p_i$, for $i = 1,\ldots, n$. Let $\mu = \sum_i p_i$. Then, for each
$\delta \ge 0$, it holds that
$$\Pr\left[\sum_i X_i > (1+\delta) \cdot \mu\right] \le
\exp\left(-\frac{\min\left(\delta,\delta^2\right)}{3} \cdot
\mu\right),$$
and,
$$\Pr\left[\sum_i X_i < (1-\delta) \cdot \mu\right] \le
\exp\left(-\frac{\delta^2}{3} \cdot
\mu\right).$$
\label{thm:chernoff}
\end{theorem}


We now show how to use Theorem~\ref{thm:chernoff}, together with
the bounds derived in Section~\ref{sec:fltwosig}, to prove
Theorem~\ref{thm:upperb_pv}.
Let $V_j$ be the number of votes to $p_j$ in
the random election, with unknown urn $i$. Then, $E[V_j] = E_i(j)
\cdot m$. We have,
\begin{align*}
\Pr[\text{the players lose}] & = \Pr[p_i \text{ did not collect more votes than any other urn}]\\
& \le \Pr\left[V_i < E[V_i] - \frac m{3M}\right] +
\sum_{\substack{j=0\\j \ne i}}^n \Pr\left[V_j > E[V_i] - \frac{2m}{3M}\right]
\end{align*}
Since $\Delta_i(j) \ge \frac{\left|i-j\right|}M$ we have $E[V_i] \ge E[V_j] + \frac{\left|i-j\right|}M\cdot m$, and
\begin{align*}
\Pr[\text{the players lose}]& \le \Pr\left[V_i < E[V_i] \left(1
  - \frac m{3 M E[V_i]}\right)\right] +
\sum_{\substack{j=0\\j \ne i}}^n
\Pr\left[V_j > E[V_j] \left(1+\frac{\left|i-j\right|}{M E[V_j]}\cdot m -
  \frac{2m}{3M E[V_j]}\right)\right]\\
& \le \Pr\left[V_i < E[V_i] \left(1
  - \frac m{3 M E[V_i]}\right)\right] +
\sum_{\substack{j=0\\j \ne i}}^n
\Pr\left[V_j > E[V_j] \left(1+\frac{\left|i-j\right|m}{3M E[V_j]}\right)\right]\\
& \le \exp\left(-\frac{m^2}{27 \, M^2 \, E[V_i]}\right) +
2\sum_{k=1}^n 
\exp\left(-\min\left\{\frac{km}{9M},\frac{k^2m^2}{27M^2E[V_j]}
\right\}\right),
\end{align*}
by $E[V_j] \le E[V_i] \le m \cdot \frac{4(n-1)}{\epsilon M}$, 
\begin{align*}
\Pr[\text{the players lose}] & \le \exp\left(-\frac{m^2 \, \epsilon \, M}{108 \, M^2 \, m \, (n-1)}\right) +
2\sum_{k=1}^n
\exp\left(-\min\left\{\frac{km}{9M},\frac{k^2m^2\epsilon M}{108\,M^2\,m\,(n-1)}
\right\}\right)\\
& \le \exp\left(-\frac{m \, \epsilon }{108 \, M \, (n-1)}\right) +
2\sum_{k=1}^n \exp\left(-\min\left\{\frac{km}{9M},\frac{k^2m\epsilon}{108\,M\,(n-1)}
\right\}\right)\\
& \le \exp\left(-\frac{m \, \epsilon }{108 \, M \, (n-1)}\right) +
2\sum_{k=1}^n \exp\left(-\frac{k^2 m \epsilon}{108 \, M \, (n-1)} \cdot \min\left\{\frac{12\,(n-1)}{k\epsilon},1
\right\}\right)\\
& \le \exp\left(-\frac{m \, \epsilon }{108 \, M \, (n-1)}\right) +
2\sum_{k=1}^n \exp\left(-k^2 \ln\frac4{\eta}\right)\\
& \le \frac{\eta}4 + 
2\sum_{k=1}^n \left(\frac{\eta}4\right)^{k} \le \frac{\eta}4 + 2 \cdot
\frac{\eta}{4-\eta} \le  \frac{\eta}4 + 2 \cdot \frac{\eta}3 < \eta.
\end{align*}
It follows that if $m = \Theta(n^3 \epsilon^{-2} \log \eta^{-1})$,
and voters apply the 
aforementioned voting scheme, the probability of
winning is at least $1 - \eta$.
\end{proof}

\section{A Connection to Proper Scoring Rules}
\label{sec:scoring}

In this section we discuss the connection between our upper bound
and the notion of a {\em proper scoring rule} 
\cite{gneiting-scoring-rules}.
We first show how to obtain a strategy for a set of voters 
in the two-signal case using proper scoring rules.\footnote{We are
grateful to Bobby Kleinberg for identifying this connection between voting
strategies and proper scoring rules.}
We then show that basing a strategy on proper scoring rules cannot
lead to an asymptotically tight result:
any voting strategy based on the functions arising
from the framework of proper scoring rules requires at least
$\Omega\left(n \epsilon^{-4}\right)$ voters.
This is weaker than the upper bound of $O(n^3 \eps^{-2} \log \eta^{-1})$
that we obtained in Section~\ref{sec:twosig} in two important respects.
First, by the pigeon-hole principle,
$\epsilon \le \frac1{n-1}$, and therefore approach from the
previous section is always at least as good as the approach 
based on proper scoring rules, and often much better.
More significantly, when $n$ is a constant, the approach via scoring rules
gives a dependence on $\eps$ of $O(\eps^{-4})$, whereas
our approach from Section~\ref{sec:twosig} gives $O(\eps^{-2})$,
which is optimal even if the group of voters could directly share
all their signals.  (In other words, even if there were just
a single voter who received all the signals.)

For our purposes in this discussion, it is not necessary to introduce
the full theory of proper scoring rules, but just to provide
a self-contained consequence of that theory.
The consequence is the following:
it is possible to construct pairs of non-negative functions $(f_0,f_1)$,
each defined over the interval $[0,1]$, with
the property that for all $z \in [0,1]$, the function
\begin{equation}
g_z(x) = z f_0(x) + (1 - z) f_1(x)
\label{eq:scoring-property}
\end{equation}
is uniquely maximized at $x = z$.
We will further assume that 
$f_0$ and $f_1$ each have continuous second derivatives,
which is true of the standard functions that arise from this theory.
This defining property of $f_0$ and $f_1$ is all we will need.

From a pair of such functions, here is how we can define a strategy
for each voter in the two-signal case.
We have a set of $n+1$ urns, where urn $i$ has
a probability $p_i$ of producing a blue ball.
We define $q_0 = \sum_i f_0(p_i)$ and $q_1 = \sum_i f_1(p_i)$, and
let $\qstar = \max(q_1,q_0)$.
Now, when a voter draws a blue ball, they vote for urn $i$
with probability proportional to $\df{f_0(p_i)}{\qstar} + \df{\qstar - q_0}{\qstar(n+1)}$;
if they draw a red ball, they vote for urn $i$
with probability proportional to $\df{f_1(p_i)}{\qstar} + \df{\qstar - q_1}{\qstar(n+1)}$.
We call this the strategy {\em induced by $f_0$ and $f_1$}.

Suppose the true urn is $t$; then the number of votes
for an urn $j$ is a random variable 
$X_j = \sum_v X_{jv}$,
where $X_{jv}$ is the indicator variable that voter $v$ votes for $j$.
With $k$ voters, we have
\begin{equation}
\Exp{X_j} = \frac{k}{\qstar} (p_t f_0(p_j) 
   + (1 - p_t) f_1(p_j)) + \frac{\left|q_1-q_0\right| k}{\qstar(n+1)}
 = \frac{k}{\qstar} g_{p_t}(p_j) + \frac{\left|q_1-q_0\right| k}{\qstar (n+1)}.
\label{eq:scoring-number-votes}
\end{equation}
By the defining property of $f_0$ and $f_1$, we see that 
$\Exp{X_j}$ is uniquely maximized at $j = t$.
Hence for a sufficiently large set of voters, the number
of votes received by urn $t$ will exceed the number received
by all other urns with high probability.

Thus, the strategy induced by any proper scoring rule will produce
the true urn with high probability when there are enough voters.
It can be viewed, in a sense, as a much simpler version of the construction
in Section~\ref{sec:twosig}, and we now show that this simpler approach
results in asymptotically larger number of voters.

\begin{theorem}
Let $f_0$ and $f_1$ be any functions with continuous second derivatives
that satisfy the defining property of proper scoring rules 
from Equation (\ref{eq:scoring-property}).
Then in order for the strategy induced by $f_0$ and $f_1$ to
identify the true urn with high probability, there must be
$\Omega(n \eps^{-4})$ voters.
\label{thm:scoring-lower-bound}
\end{theorem}
\begin{proof}
We start with a basic claim about sums of Bernoulli trials.
Let $X = \sum_{i^1}^k X_i$ be a sum of independent $0$-$1$
random variables, where $\Exp{X_i} = p_i \leq \frac12$.
The mean of $X$ is $\mu = \sum_{i=1}^k p_i$.
Then with constant probability, $X$ will deviate
by at least a constant multiple of $\sqrt{\var{X}}$ from $\mu$.
More concretely, there are absolute constants $\alpha > 0$ and
$\beta > 0$ so that with probability at least $\alpha$, 
we have $X < \mu - \beta \sqrt{\var{X}}$.
Now, since $$\var{X_i} = p_i (1 - p_i) \geq p_i / 2,$$ 
we have 
$$\var{X} = \sum_{i=1}^k \var{X_i} \geq \sum_{i=1}^k p_i/2 \geq \mu/2.$$
Now, for a given $\delta > 0$, suppose we have
$\mu < \beta^2 / (2 \delta^2)$.
Then equivalently, 
$\delta < \beta / \sqrt{2 \mu}$, so
$$\delta \mu < \beta \sqrt{\mu/2} \leq \beta \sqrt{\var{X}}.$$
Hence with probability at least $\alpha > 0$, we have
$X < (1 - \delta) \mu$.
It follows that in order to ensure $X \geq (1 - \delta) \mu$ with
probability going to $1$, we must have
$\mu \geq \Omega(1 / \delta^2)$.

Now, recall that there are $k$ voters, and 
consider the voting strategy induced by
the functions $f_0$ and $f_1$.
Since the first derivatives ${f_0'}$ and ${f_1'}$ are
continuous functions defined over the compact set $[0,1]$, 
there is a constant $c_1$ such that $|f_0'(x)|, |f_1'(x)| \leq c_1$
for all $x \in [0,1]$.
For the same reason, 
there is a constant $c_2$ such that $|f_0''(x)|, |f_1''(x)| \leq c_2$
for all $x \in [0,1]$.
Using the bound on the first derivative, for any $\gamma > 0$ 
we can find an interval $[u,v] \subseteq [0,1]$ 
such that the following hold:
(i) $\ds{d = \inf_{x \in [u,v]} \min(f_0(x),f_1(x)) > 0}$, 
(ii) $v/u < 1 + \gamma$, 
(iii) $(1 - u)/(1 - v) < 1 + \gamma$, 
$$\mbox{(iv)} ~ 
\sup_{x,y \in [u,v]} \frac{f_0(y)}{f_0(x)} < 1 + \gamma,
~ \mbox{and} ~
\mbox{(v)} ~
\sup_{x,y \in [u,v]} \frac{f_1(y)}{f_1(x)} < 1 + \gamma.$$
It follows that if our probabilities $p_0, p_1, \ldots, p_n$
all lie in this interval $[u,v]$, then 
$$\Exp{X_j} \in \left[\df{(1 - \gamma_1) k}{n},\df{(1 + \gamma_1)k}{n}\right]$$
for a constant $\gamma_1$ that goes to $0$ with $\gamma$.
Also, we have $\qstar \geq dn$.

Now, for any $\eps > 0$, we choose 
$p_0 \leq p_1 \leq \dots \leq p_n \in [u,v]$
such that $p_{j+1} - p_j = \eps$ for each $j$.
Let $t$ be the true urn, and let 
$$h(x) = \frac{k}{\qstar} (p_t f_0(x) 
   + (1 - p_t) f_1(x)) + \frac{\left|q_1 - q_0\right| k}{\qstar(n+1)}.$$
Notice that $\Exp{X_j} = h(p_j)$.
Now, Taylor's Theorem implies that for some $w \in [p_t,p_{t+1}]$, we have
$$h(p_{t+1}) = h(p_t) + (p_{t+1} - p_t) h'(p_t) 
   + \frac12 (p_{t+1} - p_t)^2 h''(w).$$
Since $h(x)$ has its global maximum at $x = p_t$, we have $h'(p_t) = 0$.
Moreover, since $|f_0''(x)|, |f_1''(x)| \leq c_2$
for all $x \in [0,1]$, we have
$$h''(w) = \frac{k}{\qstar} (p_t f_0''(x) + (1 - p_t) f_1''(x))
\leq \frac{k c_2}{d n}.$$
Writing $p_{t+1} - p_t = \eps$, we have
$$h(p_{t+1}) \geq h(p_t) - \frac{k c_2}{2 d n} \eps^2$$
Since $\Exp{X_j} = h(p_j)$, for all $j$, this implies
$$\Exp{X_{t+1}} \geq \Exp{X_t} - \frac{k c_2}{2 d n} \eps^2.$$
Since $\Exp{X_t} \geq \df{(1 - \gamma_1) k}{n}$, this implies
that $\Exp{X_{t+1}} \geq (1 - \delta) \Exp{X_t}$, where
$\delta = c_3 \eps^2$ and $c_3 = \df{c_2}{2 (1 - \gamma_1) d}$.

Now, using our initial fact about sums of Bernoulli trials, 
we must have $\Exp{X_t} \geq \Omega(1 / \delta^{2})$ in order
for $X_t$ to have a high probability of exceeding $(1 - \delta) \Exp{X_t}$.
Since $\Exp{X_t} \leq \df{(1 + \gamma_1) k}{n}$, this requires
$$\df{(1 + \gamma_1) k}{n} \geq \df{d_2}{c_3^{2} \eps^{4}}$$
for a constant $d_2 > 0$, 
and hence 
$$k \geq \frac{d_2 n}{(1 + \gamma_1) c_3^2 \eps^4}.$$
\end{proof}

\comment{We provide a proof of this theorem in the appendix.
As a sketch of how the proof works, suppose the true urn is $t$
and that there are $k$ voters, and
consider the function 
$$h(x) = \frac{k}{\qstar } (p_t f_0(x) 
   + (1 - p_t) f_1(x)) + \frac{\left|q_1 - q_0\right| k}{\qstar (n+1)}.$$
As we saw in Equation (\ref{eq:scoring-number-votes}), the expected number
of votes received by each urn $j$ is $h(p_j)$.
Now, since $h(x)$ has its global maximum at $x = p_t$, 
Taylor's Theorem implies that it looks like an inverted parabola
with zero derivative at $x = p_t$.
From this, and the fact that the votes must be spread out across
a linear (in $n$) number of urns, we can show that 
$h(p_{t+1}) \geq h(p_t) - a_0 \eps^2 k / n$ for a constant $a_0$.
This is the crucial point: the global maximum at $x = p_t$
implies that the number of votes to nearby urns falls off
only quadratically in the gap between their probabilities.
Since $h(p_t) = a_1 / n$ for a constant $a_1$, it follows that
the expected number of votes received by urn $t+1$ is only a
factor of $(1 - a_2 \eps^2)$ less than the expected number of
votes received by urn $t$, for a constant $a_2$.
Hence if the number of voters does not exceed a constant times
$n \eps^{-4}$, then there is a constant probability that
urn $t+1$ will receive more votes than urn $t$,
resulting in the wrong answer.}

\section{A Tight Lower Bound for Two Signals}


In this section we give a tight lower bound that confirms the
optimality of the voting scheme for two signals presented in
Section~\ref{sec:twosig}. We will start by introducing a class of
instances.
We will then prove a
combinatorial lemma on how certain parameters of any (asymmetric) 
voting system for these instances have to behave, and we use
the lemma to prove the lower bound.

\medskip

We start by defining the lower bound class of  
instances $\mathcal{I}(n,\epsilon)$, for
any $n \ge 2$ and $\epsilon \le \frac1{n-1}$. The $n$ urns in
$\mathcal{I}(n,\epsilon)$ are such that $p_i =
\dfrac{1-\epsilon(n-1)}2+(i-1)\epsilon$, for $i = 1, \ldots, n$. Then $0
\le p_1 \le
p_2 \le \cdots \le p_{n} \le 1$.


Each voter $t$ is defined by two probability distributions $(R_{1,t},
R_{2,t}, \ldots, R_{n,t})$, $(B_{1,t}, B_{2,t}, \ldots, B_{n,t})$: if
she draws a red (resp., blue) ball she will vote for urn $i$ with
probability $R_{i,t}$ (resp., $B_{i,t}$).

Given a voting scheme for $m$ voters (that is, $2m$ probability
vectors $(R_{i,t}), (B_{i,t})$), we define $B_i = m^{-1} \cdot
\sum_{t=1}^m B_{i,t}$ and $R_i = m^{-1} \cdot \sum_{t=1}^m R_{i,t}$,
for $i = 1, \ldots, m$. Thus the expected number of votes $E_i(j)$ to urn $j$,
if $i$ is the correct urn, will be equal to
$m \cdot E_i(j) = m \cdot \left(p_i \cdot B_{j} + (1-p_i) \cdot
R_j\right).$
We also define $\Delta_i(j) = E_i(i) - E_i(j)$ to be the
expected difference between the number of votes to $i$ and $j$, if $i$
is the correct urn, averaged over the $m$ voters.

We say that a voting scheme is {\em proper} 
if $\Delta_i(j) \ge 0$, for each $i,j$.
The challenge in proving the lower bound lies in the fact that
proper voting schemes can succeed in identifying the correct urn for
what seem to be a variety of different reasons, and so we need to
find a common property they have which implies that the correct
urn only ``narrowly'' wins the election over other urns with
very similar distributions.
This is the content of the following lemma.

\begin{lemma}\label{lemma:BR}
Let $n$ and $\epsilon$ be $n \ge 10$ and $\epsilon \le \frac1{n-1}$.
Then all proper voting schemes for $\mathcal{I}(n,\epsilon)$
satisfy:\begin{itemize}
\item[(a)] $B_1 \le B_2
\le \cdots \le B_{n} \le \frac{9}{n}$ and 
$\frac{9}{n} \ge R_1 \ge R_2 \ge \cdots \ge R_{n}$;
\item[(b)] $E_i(i) \le \frac{9}{n}$, for $i = 1, \ldots, n$.
\item[(c)] There exists a set $S \subseteq [n]$ and $\iota \in
  \{-1,+1\}$, with
$\left|S\right| \ge \frac{n}{4}-3\ln n -14$, such that
  for each $i \in S$, we have
$$\max\left(\left|R_i - R_{i+\iota}\right|, \left|B_i - B_{i+\iota}\right|\right) < e^{\frac72} \cdot \frac{\sqrt{\left|R_{i} -
    B_{i}\right|}}{n^{\frac32}},$$
and
$$\Delta_i(i+\iota), \Delta_{i+\iota}(i) \le 2 e^{\frac72}
  \cdot \epsilon \cdot \frac{\sqrt{\left|R_{i} -
    B_{i}\right|}}{n^{\frac32}}.$$
\end{itemize}
\end{lemma}
The crux of the lemma is to show that for many pairs of 
urns $i, i + \iota$, the election will be very ``close'':
if $i$ is the correct urn, it does not win the election 
by a large margin over $i+\iota$ in expectation (and
vice versa). 
The lemma shows further that, averaged over the
voters, the difference between the probability of voting 
for $i$ given a red (resp., blue)
ball and the probability of voting for $i+\iota$ given the same color
is small. 

This upper bound is crucial for the proof of the lower-bound theorem, 
stated next: we
will show that --- even if we only cared about urns $i, i+ \iota$ ---
the variance of a voter's choice can be lower-bounded by
$\Omega(\left|R_{i,t} - B_{i,t}\right|)$. This, assuming that the
total variance is at least some constant, will allow us to apply
an anti-concentration inequality to show that 
the expected margin $\Delta_i(i + \iota)$ of urn $i$ over urn $i + \iota$
will be surpassed by $\Theta\left(\ln \eta^{-1}\right)$ 
standard deviations of the number of votes to urn
$i$ and $i+\iota$.  It will follow that with probability $\Omega(\eta)$ the
election will be won by the wrong urn. 
Again, this argument requires that the variance be at least some 
sufficiently large constant; 
if the variance is actually
smaller than this constant, we will use a different argument showing
that the voting system is sufficiently ``inflexible'' that if
urn $i$ wins when it is correct, the same pattern of votes is
likely to also arise --- favoring $i$ --- when $i + \iota$ is actually correct.

\begin{proof}
We will first show that in a proper voting scheme, for each $i < j$
it holds that $B_i \le B_j$ and $R_i \ge R_j$. This implies $B_0 \le B_1
\le \cdots \le B_{n-1}$ and $R_0 \ge R_1 \ge \cdots \ge
R_{n-1}$. By contradiction,
\begin{itemize}
\item if $B_i < B_j$ and $R_i < R_j$, then $E_{k}(j) >
  E_{k}(i)$ for each $k$: in particular for $k = i$, which would give $\Delta_i(j) < 0$, contradicting the properness of the voting scheme;
\item the same argument gives a contradiction if $B_i > B_j$ and $R_i
  > R_j$ (choosing $k = j$);
\item finally, assume $B_i > B_j$, $R_i < R_j$, and that $E_{i}(i) >
  E_{i}(j)$, $E_{j}(i) < E_{j}(j)$; then,
\begin{align*}
E_{i}(i) - E_{j}(i) & > E_{i}(j) - E_{j}(j)\\
(p_i - p_j) B_i + (p_j - p_i) R_i & > (p_i - p_j) B_j + (p_j - p_i)
R_j\\
(p_i - p_j) (B_i-B_j) & > (p_j - p_i)(R_j - R_i),
\end{align*}
then, by $p_i < p_j$, we have
$$  (R_j - R_i) + (B_i - B_j) < 0,$$
which is impossible since the left-hand side
is positive by $B_i > B_j$ and $R_j > R_i$.
\end{itemize}
It follows that $B_i \le B_j$ and $R_i \ge R_j$. We now show that $B_n
\le \frac9n$ (resp., $R_1 \le \frac9n$). Since $\{B_i\}_{i=1}^n$ and $\{R_i\}_{i=1}^n$ are probability
distributions, one has that $\left|\left\{i \mid B_i+R_i \le \frac{4}n\right\}\right| \ge \frac n2$, for otherwise
$$2 = \sum_{i=1}^n  B_i + \sum_{i=1}^n
 R_i \ge \sum_{\substack{i=1\\B_i +
    R_i > \frac4n}}^n E_i(i) > \frac n2 \cdot \frac{4}n
\ge 2.$$
Since $p_{\lfloor\frac{n+1}2\rfloor} \ge \frac{1-\epsilon}2$ (resp.,
$p_{\lceil\frac{n+1}2\rceil} \le \frac{1+\epsilon}2$), it follows that there exists some $i$ such that $B_i + R_i \le
\frac4n$ and $p_i \ge \frac{1-\epsilon}2 $ (resp., $p_i \le \frac{1+\epsilon}2$). Observe
that $E_i(i) \le \frac4n$. Now, by contradiction, let $B_n > \frac9n$
($R_1 > \frac9n$); by $n \ge 10$ we have $\epsilon \le \frac1{n-1} \le
\frac19$; therefore
$p_i \ge \frac49$ ($p_i \le \frac59$) one
has $E_i(n) \ge p_i \cdot B_n > \frac4n$ ($E_i(1) \ge p_i \cdot
R_1 > \frac4n$). It follows that $\Delta_i(n)$ ($\Delta_i(1)$) is
negative, contradicting properness.

\comment{
\medskip

We now move on to the rest of the statement.

Let $S_r \subseteq[n]$ (resp., $S_b \subseteq [n]$) be the set of indices $i$ for which $R_i <
\frac{10}{n}$ (resp., $B_i < \frac{10}n$). Since $R_i \ge 0$, for each $i = 0,\ldots, n-1$, and
$\sum_{i=0}^{n-1} R_i = 1$, we have that $\left|[n] - S_r\right| \le
\frac{n}{10}$ (for otherwise the sum $\sum_{i \in ([n] - S_r)} R_i$
would be larger than $\sum_{i \in [n]} R_i$, a contradiction). It
follows that $\left|S_r\right| \ge \frac9{10} n$.

 Analogously, we have  $\left|S_b\right| \ge
\frac9{10} n$.

\medskip

Let $S'_r \subseteq S_r \cap [n-1]$ (resp., $S'_b \subseteq S_b \cap [n-1]$) be the set of indices $i$ for
which $i, i+ 1 \in S_r$ (resp., $i, i+1 \in S_b$). Since each index $i
\in [n-1]$ missing in $S_r$ is
part of at most two couples of consecutive elements (that is,
$(i-1,i)$ and $(i,i+1)$), we have that
$\left|S'_r\right| \ge (n-1) - 2\left|[n] - S_r \right| \ge n-1 -
2\frac{n}{10} \ge \frac45n - 1$.

 The same reasoning gives
$\left|S'_b\right| \ge \frac45 n - 1$.}

\medskip

We define $\delta_i = \left|R_i - B_i\right|$. Then $\delta_n \le
B_n\le \frac9n$ and $\delta_1 \le R_1 \le \frac9n$.
Let $k_r$ be the largest integer such that $R_{k_r+1}
\ge B_{k_r+1}$, and $k_b$ be the largest integer such that $B_{n-k_b}
\ge R_{n-k_b}$. Observe that (a) if $i \le k_r$ then $\delta_i \ge
\delta_{i+1}$, (b) if 
$i \ge n-k_b+1$ then $\delta_{i} \ge \delta_{i-1}$, and (c) $k_r + k_b
\ge n-2$. By (c) at least one
of $k_r$ and $k_b$ has to be at least $\frac n2 - 1$. We let $S_R$ and $S_B$ be
\begin{align*}
S_R &= \left\{i \mid R_{i+1} \ge B_{i+1} \wedge \max\left(R_i - R_{i+1}, B_{i+1} -
B_i\right) < e^{\nicefrac72} \cdot \frac{\sqrt{R_{i} -
    B_{i}}}{n^{\nicefrac32}}\right\},\\
S_B &= \left\{i \mid B_{i-1} \ge R_{i-1} \wedge \max\left(R_{i-1} - R_{i}, B_{i} -
B_{i-1}\right) < e^{\nicefrac72} \cdot \frac{\sqrt{B_i -
    R_{i}}}{n^{\nicefrac32}}\right\}.
\end{align*}
Then $S_R \subseteq \{1, 2, \ldots, k_r\}$ and $S_B \subseteq \{n-k_b+1,
n-k_b+2, \ldots, n-1,n\}$.
 We consider two cases:\begin{itemize}
\item suppose $k_r \ge \frac n2-1$.  We relabel the
  element in $\overline{S}_R = [k_r] - S_R$, using $r = \left|\overline{S}_R\right|$:
$$\overline{S}_R = \{i_1, i_2, \ldots, i_{r}\},$$
with $i_1 < i_2 < \cdots < i_{r}$. We have $\delta_1 \ge \delta_2 \ge
  \cdots \ge \delta_{k_r} \ge \delta_{k_r+1}$. Then, for $1 \le t \le
  r - 1$,
\begin{equation}\label{delta_eqn}
\delta_{i_{t+1}} \le \delta_{i_t+1} \le \delta_{i_t} -
e^{\nicefrac72} \cdot \frac{\sqrt{R_{i_t} - B_{i_t}}}{n^{\nicefrac32}} =\delta_{i_t} -
e^{\nicefrac72} \cdot \frac{\sqrt{\delta_{i_t}}}{n^{\nicefrac32}} =
\delta_{i_t} \cdot \left(1 - \sqrt{\frac{e^7}{n^3 \cdot
    \delta_{i_t}}}\right),
\end{equation}
and $\delta_{i_1} \le \frac9n \le \frac{e^{3}}n$; we define $\alpha_k
= e^{3-k}$ so that $\delta_{i_1} \le \alpha_0 \cdot n^{-1}$.  Let
$\ell_0 = 1$ and $\ell_k = \left\lceil n \cdot e^{-\nicefrac{(k+3)}2}\right\rceil$, for $k \ge 1$. We also let $L(k) = \sum_{j=0}^k
\ell_j$. We will
show by induction on $k$ that $\delta_{i_{L(k)}} \le \alpha_k \cdot n^{-1} = e^{3-k} \cdot
n^{-1}$. The case $k = 0$ has already been verified. We assume $k \ge
1$. Then,
\begin{align*}
\delta_{i_{L(k)}} &\le \delta_{i_{L(k)-1}} \cdot \left(1-
\sqrt{\frac{e^7}{n^3\cdot \delta_{i_{L(k)-1}}}}\right) \le \delta_{i_{L(k-1)}} \cdot \left(1-
\sqrt{\frac{e^7}{n^3\cdot \delta_{i_{L(k-1)}}}}\right)^{\ell_k}\\
&\le \delta_{i_{L(k)-1}} \cdot \left(1-
\sqrt{\frac{e^{k+3}}{n^2}}\right)^{\ell_k} = \delta_{i_{L(k-1)}} \cdot
\left(1- \frac{e^{\frac{k+3}2}}n\right)^{\ell_k}\\
&= \delta_{i_{L(k-1)}} \cdot
\left(1-
\frac{e^{\frac{k+3}2}}n\right)^{\left\lceil\frac{n}{e^{\nicefrac{(k+3)}2}}\right\rceil}
\le \delta_{i_{L(k-1)}} \cdot e^{-1} \le e^{3-k} \cdot n^{-1}.
\end{align*}
Now, if $k \ge 11 + 3\ln n$, we have $\delta_{i_{L(k)}} \le n^3 \cdot
  e^{-8}$; by~(\ref{delta_eqn}), we would then get
$\delta_{i_{L(k)+1}} \le \delta_{i_{L(k)}} \cdot \left(1 - e\right) <
  0$ --- since $\delta_{i_r} \ge 0$, by $i_r \le k_r$, this implies that $r =
  \left|\overline{S}_R\right| < L\left(\left\lceil11+3\ln
  n\right\rceil\right)$.
We now upper bound $L(k)$ to get an upper bound on
$r = \left|\overline{S}_R\right|$:
\begin{align*}
L(k) &= \sum_{j=0}^k \ell_j = 1 + \sum_{j=1}^k \left\lceil n \cdot
e^{-\frac{k+3}2}\right\rceil \\
&= k+1 + n \cdot \sum_{j=1}^k e^{-\frac{k+3}2} \le k+1 + n  \cdot
e^{-\nicefrac52} \cdot \sum_{j=0}^{\infty} e^{-\nicefrac k2}\\
& = k + 1
+ n \cdot e^{-\nicefrac52} \cdot \frac1{1-e^{-\nicefrac12}} = k + 1
+  \frac n{e^{\nicefrac52} - e^2} \le \frac n4 + k+1.
\end{align*}
It follows that $r = \left|\overline{S}_R\right| <
L\left(\left\lceil11+3\ln n\right\rceil\right) \le \frac  n4 +
\left\lceil11+3\ln n\right\rceil + 1 \le \frac n4 + 3\ln n + 13$, and
therefore
$$\left|S_R\right| \ge k_r -  \frac n4 - 3\ln n - 13 \ge \frac n4 -
3\ln n - 14.$$

\item Otherwise, $k_r < \frac n2-1$ and therefore $k_b > \frac n2
  -1$. A proof similar to the previous case gives
$$\left|S_B\right| \ge k_b -  \frac n4 - 3\ln n - 13 \ge \frac n4 -
3\ln n - 14.$$
\end{itemize}


 Therefore, at least one of $S_R$ and $S_R$ has cardinality at
least $\frac{n}{4}-3\ln n-14$. If $S_R$ is the largest one we pick $\iota = 1$
and $S = S_R$. Otherwise, we pick $\iota = -1$ and $S = S_B$. Observe
that the choice satisfies the first requirement of point (c) in the statement.

\medskip

We now prove the second requirement of point (c). Let $i$ be an element of $S$,
$\beta = B_{i} - B_{i+\iota}$, and $\rho = R_i - R_{i + \iota}$. Then,
$\left|\beta\right| = -\iota\beta $ and $\left|\rho\right| = \iota\rho$. Also,
$$\left|\beta\right|, \left|\rho\right| \le e^{\frac72} \cdot \frac{\sqrt{\left|R_{i} -
    B_{i}\right|}}{n^{\frac32}}.$$
Recall that $\Delta_{i}(i+\iota) = E_{i}(i) - E_{i}(i+\iota) = \beta p_{i} + \rho
(1-p_{i})$ and $\Delta_{i+\iota}(i) = E_{i+\iota}(i+\iota) - E_{i+1}(i) = -\beta
p_{i+\iota}-\rho (1-p_{i+\iota})$. Suppose that at least one of $\Delta_{i}(i+\iota)$ and
$\Delta_{i+\iota}(i)$ is larger than $2e^{\frac72}\epsilon \frac{\sqrt{\left|R_{i} -
    B_{i}\right|}}{n^{\frac32}}$. By the properness of the voting system, we
would have:
\begin{align*}
\Delta_{i}(i+\iota) + \Delta_{i+\iota}(i) & > 2e^{\frac72} \cdot \epsilon \cdot \frac{\sqrt{\left|R_{i} -
    B_{i}\right|}}{n^{\frac32}}\\
\beta (p_{i} - p_{i+\iota}) + \rho
(p_{i+\iota}-p_{i}) & > 2e^{\frac72} \cdot \epsilon \cdot \frac{\sqrt{\left|R_{i} -
    B_{i}\right|}}{n^{\frac32}},
\end{align*}
by the definition of the instance, we have that $p_{i+\iota}-p_{i} =
\iota \cdot \epsilon$, therefore we would have
$$\left|\beta\right|  + \left|\rho\right|  > 2e^{\frac72} \cdot \frac{\sqrt{\left|R_{i} -
    B_{i}\right|}}{n^{\frac32}}$$
which would imply that at least one of $\left|\beta\right|$ and $\left|\rho\right|$ is larger
than $e^{\frac72} \cdot \frac{\sqrt{\left|R_{i} -
    B_{i}\right|}}{n^{\frac32}}$, a contradiction. It follows that both
$\Delta_{i}(i+\iota)$ and $\Delta_{i+\iota}(i)$ are less than or equal
$$\Delta_i(i+\iota), \Delta_{i+\iota}(i) \le 2e^{\frac72} \cdot \epsilon \cdot \frac{\sqrt{\left|R_{i} -
    B_{i}\right|}}{n^{\frac32}}.$$

\end{proof}

\begin{theorem}\label{thm:mainLB}
There exists a positive constant $H$ such that for any $\eta < H$, one
has that any voting scheme for
$\mathcal{I}(n,\epsilon)$, with $n \ge
120$ and $\epsilon \le \frac1{11(n-1)}$ using at most $O\left(\frac{n^3}{\epsilon^2} \log
\frac1{\eta}\right)$ voters will fail to win the election with
probability $\Omega(\eta)$.
\end{theorem}
\begin{proof}
Take any asymmetric voting scheme for $\mathcal{I}(n,\epsilon)$ with
$m$ voters --- that is, a sequence of $m$ vectors
$(R_{1,t},\ldots, R_{n,t})$ and $(B_{1,t},\ldots, B_{n,t})$, for $1
\le t \le m$, such that the probability that the $t$th voter votes for
the $i$th urn if she draws a blue (resp., red) ball is $B_{i,t}$
(resp., $R_{i,t}$). Let $B_i = m^{-1} \cdot \sum_{t=1}^m B_{i,t}$ and
$R_i = m^{-1} \cdot \sum_{t=1}^m R_{i,t}$.

If the voting scheme is improper, then by definition there exists
$i,j$ such that $\Delta_i(j) < 0$. Otherwise, by $n \ge 120$, one has $\frac n4 - 3\ln n - 14 \ge 1$, and
 by Lemma~\ref{lemma:BR},
there will exist two urns $i$ and $j\in\{i-1,i+1\}$ such that $\Delta_i(j) \le
2e^{\frac72}\epsilon \frac{\sqrt{\left|R_i-B_i\right|}}{n^\frac32}$.

Given $i,j$, we define the head-to-head $(i,j)$-voting process as
follows; for each voter $t$, the random variable $X_t = X_t(i,j)$ will
be defined as
$$X_t = \left\{\begin{array}{cl}
1 & \text{if voter } t \text{ votes for urn } j \text{, given that the unknown urn
  is } i,\\
\nicefrac12 & \text{if voter } t \text{ does not vote for urns } i \text{ or } j
\text{, given that the unknown urn is } i,\\
0 & \text{if voter } t \text{ votes for the unknown urn } i.
\end{array}\right.$$
Observe that $X = \sum_{t=1}^m X_t \ge \frac m2$ iff the number of votes
to urn $j$ is not smaller than the number of votes to the right urn $i$. In this
case, the voters will lose the election. Furthermore,
$$E[X] = \frac m2 - \frac m2 \cdot \Delta_i(j).$$

Since $X$ is the sum of independent random variables, we have that
$\var[X] = \sum_{t=1}^m \var[X_t]$; by $\epsilon \le \frac1{11}{(n-1)}$,
we have that $\frac5{11} \le p_1 \le p_2 \le \cdots \le p_n \le
\frac6{11}$. We will use that $p_i,1-p_i \ge \frac14$ for each $i$, to
lower-bound the variance of $X_t$:
\begin{align*}
\var[X_t] &= (p_i \cdot B_{i,t} + (1-p_i) \cdot R_{i,t}) \cdot \left(0-
E[X_t]\right)^2  + (p_i \cdot B_{j,t} + (1-p_i) \cdot R_{j,t}) \cdot
\left(1-E[X_t]\right)^2 \\
&+(p_i \cdot (1-B_{i,t}-B_{j,t}) 
 + (1-p_i) \cdot (1-R_{i,t}-R_{j,t})) \cdot \left(\frac12-E[X_t]\right)^2.
\end{align*}
We consider two cases:\begin{itemize}
\item if $E[X_t] \ge \frac14$, then
$$\var[X_t] \ge (p_i B_{i,t} + (1-p_i)  R_{i,t}) \cdot \left(0-
E[X_t]\right)^2 \ge \frac{R_{i,t} + B_{i,t}}4 \cdot \frac1{16} \ge
 \frac{\left|R_{i,t} - B_{i,t}\right|}{64}.$$
\item if $E[X_t] < \frac14$, we have
\begin{align*}
\var[X_t] &\ge (p_i B_{j,t} + (1-p_i)  R_{j,t}) 
\left(1-E[X_t]\right)^2 \\
&+(p_i  (1-B_{i,t}-B_{j,t})
 + (1-p_i)  (1-R_{i,t}-R_{j,t})) \left(\frac12-E[X_t]\right)^2\\
& \ge \frac{p_i B_{j,t} + (1-p_i)  R_{j,t}}{16} + \frac{p_i  (1-B_{i,t}-B_{j,t})
 + (1-p_i)  (1-R_{i,t}-R_{j,t})}{16}\\
& = \frac{p_i  (1-B_{i,t})
 + (1-p_i)  (1-R_{i,t})}{16}.
\end{align*}
The latter is equal to both $\frac1{16}p_i(R_{i,t}-B_{i,t}) + (1-R_{i,t})$ and
$\frac1{16}(1-p_i)(B_{i,t}-R_{i,t})+(1-B_{i,t})$; we can therefore get a lower
bound of
$$\var[X_t] \ge \frac1{16} \cdot \min(p_i,1-p_i) \cdot \left|R_{i,t} -
B_{i,t}\right| \ge \frac{\left|R_{i,t} - B_{i,t}\right|}{64}$$
\end{itemize}
It follows that $\var[X] = \sum_{t=1}^m \var[X_t] \ge \frac1{64} \cdot
\sum_{t=1}^m \left|R_{i,t} - B_{i,t}\right|$.

\medskip

Recall that $m \cdot R_i = \sum_{t=1}^m R_{i,t}$ and $m \cdot B_i =
\sum_{t=1}^m B_{i,t}$. Suppose $R_i \ge B_i$; then
$$m \cdot \left|R_i - B_i\right| = m \cdot (R_i - B_i) = m \cdot
\sum_{t=1}^m (R_{i,t} - B_{i,t}) \le m \cdot
\sum_{t=1}^m \left|R_{i,t} - B_{i,t}\right| \le 64 \cdot \var[X].$$
If, on the other hand, $B_i > R_i$, we have
$$m \cdot \left|R_i - B_i\right| = m \cdot (B_i - R_i) = m \cdot
\sum_{t=1}^m (B_{i,t} - R_{i,t}) \le m \cdot
\sum_{t=1}^m \left|B_{i,t} - R_{i,t}\right| \le 64 \cdot \var[X].$$
Therefore, in any case, we have $\var[X] \ge \frac{m \left|R_i -
  B_i\right|}{64}$.

\medskip

 We now give a different lower bound on $\var[X]$,
that we will use to deal with the case of very small variance $\var[X]$.
Let $p_{1,t}, p_{\nicefrac12,t}, p_{0,t}$ be, respectively, the
probabilities that $X_t = 1, X_t = \frac12$ and $X_t = 0$. Then,
$E[X_t] = p_{1,t} + \frac12 \cdot p_{\nicefrac12,t}$, and
\begin{align*}
\var[X_t] & = p_{1,t} \cdot \left(E[X_t]-1\right)^2 +
p_{\nicefrac12,t} \cdot \left(E[X_t] - \frac12\right)^2 + p_{0,t}
\cdot
\left(E[X_t]\right)^2
\end{align*}
We consider three cases:\begin{itemize}
\item if $p_{1,t} = \max\left(p_{1,t},
  p_{\nicefrac12,t},p_{0,t}\right) \ge \frac13$ then if $E[X_t] \le \frac34$ we
  have
$$\var[X_t] \ge p_{1,t} \cdot \left(E[X_t]-1\right)^2 \ge \frac13
  \cdot \frac1{4^2} =
\frac1{48} \ge \frac{1-p_{1,t}}{48}.$$
If instead $E[X_t] > \frac34$, then
$$\var[X_t] \ge 
 p_{\nicefrac12,t} \left(E[X_t] - \frac12\right)^2 + p_{0,t} 
\left(E[X_t]\right)^2 > \frac{p_{\nicefrac12,t}}{4^2}+
 \frac{p_{0,t}}{4^2} \ge \frac{1-p_{1,t}}{16}.$$
\item If $p_{0,t} = \max\left(p_{1,t},
  p_{\nicefrac12,t},p_{0,t}\right) \ge \frac13$ then we employ a similar
  approach. If $E[X_t] \ge \frac14$ we
  have
$$\var[X_t] \ge p_{0,t} \cdot \left(E[X_t]\right)^2 \ge \frac13 \cdot \frac1{4^2} =
\frac1{48} \ge \frac{1-p_{0,t}}{48}.$$
If $E[X_t] < \frac14$, then
$$\var[X_t] \ge 
 p_{\nicefrac12,t} \left(E[X_t] - \frac12\right)^2 + p_{1,t} 
\left(E[X_t] - 1\right)^2 \ge \frac{p_{\nicefrac12,t}}{4^2}+
 \frac{p_{1,t}}{4^2} \ge \frac{1-p_{0,t}}{16}.$$
\item If $p_{\nicefrac12,t} = \max\left(p_{1,t},
  p_{\nicefrac12,t},p_{0,t}\right) \ge \frac13$, then $\frac16 \le
  E[X_t] \le \frac56$, and
$$\var[X_t] \ge 
 p_{1,t} \left(E[X_t] - 1\right)^2 + p_{0,t} 
\left(E[X_t]\right)^2 \ge \frac{p_{1,t}}{6^2}+
 \frac{p_{0,t}}{6^2} \ge \frac{1-p_{\nicefrac12,t}}{36}.$$
\end{itemize}
In each of the three cases, we had $\var[X_t] \ge
\frac{1-\max\left(p_{1,t}, p_{\nicefrac12,t},p_{0,t}\right)}{48}$, and
therefore
$$\var[X] = \sum_{t=1}^m \var[X_t] \ge \frac 1{48}\cdot\sum_{t=1}^m \left(1-\max\left(p_{1,t},
p_{\nicefrac12,t},p_{0,t}\right)\right).$$

Let us now assume that $\var[X] \le \frac1{72} \cdot \log_5
\frac1{\eta}$. We will deal with the case $\var[X] > \frac1{72}  \cdot
\log_5 \frac1{\eta}$ later.
 The
previous inequality then implies
$$\sum_{t=1}^m \left(1-\max\left(p_{1,t},
p_{\nicefrac12,t},p_{0,t}\right)\right) \le \frac23 \log_5 \frac1{\eta}.$$

 Recall that $X = X(i,j) \ge \frac m2$ iff the unknown urn $i$ gets
at most as many votes as $j$ (and therefore the election is lost). In the following we will
also consider $X' = X(j,i)$; we have that $X' \le \frac m2$ iff
urn $i$ gets at least as many votes as the unknown urn $j$ (this also
implies that the election is lost).

 Observe that, since $\frac5{11} \le p_1 \le
p_2 \le \cdots \le p_n \le \frac6{11}$, no matter what the unknown urn
is, the probability that any specific voter votes for any specific urn
changes by a constant factor (between $\frac56$ and $\frac65$) if one changes the
unknown urn.

We now show that, given that
$\var[x] \le \frac1{72}\cdot \log_5 \eta^{-1}$, then with probability at least
$\nicefrac{\eta}9$ each voter will vote according to its maximum probability
choice: that is $$\Pr[ \forall t, X_t \text{ equals the value } x_t
  \text{ that maximizes } \Pr[X_t = x_t]] \ge \frac{\eta}9.$$
If these choices let an urn different from $i$ win the election, we
have proven the theorem. Otherwise, we show that --- if we exchange
the unknown urn with any other urn $k$ --- then still with probability
at least $\nicefrac{\eta}{25}$ each voter $t$ will vote for the same urn $x_t$;
implying either a tie at the top, or that $i$ (which would then not be
the correct urn anymore) would will the election.

We let $s_t$ denote the sum of the two minimum probabilities in
$\{p_{1,t}, p_{\nicefrac12,t},p_{0,t}\}$; that is $s_t = 1-\max\left(p_{1,t},
p_{\nicefrac12,t},p_{0,t}\right)$. Observe that $s_t \le 1- \frac13 = \frac23$ for
each $t$. If we define $s = \sum_{t=1}^m s_t$, we also have $s \le \frac23
\log_5 \eta^{-1}$.

We have,
$$\Pr\left[ \forall t, X_t \text{ equals the value } x_t
  \text{ that maximizes } \Pr[X_t = x_t]\right] = \prod_{t=1}^m (1-s_t).$$

We now lower-bound the product, using the following greedy algorithm:
take one of the largest $s_{t'} < \frac23$, and one of the smallest
$s_t > 0$, with $s_t \ne s_{t'}$. Then  move
$x = \min\left(\frac23 - s_{t'}, s_{t}\right) > 0$ mass from $s_t$ to
$s_{t'}$. Observe that the sum $s$ of the $s_t$'s remains constant
throughout the process; furthermore the product $\prod_{m=1}^t s_t$
decreases: indeed, consider the product of $s_t \cdot s_{t'}$ before
and after the change --- we can disregard the rest since it remains
constant. Let $s_t, s_{t'}$ be the two values before the
change, and $s_t-x, s_{t'}+x$ be the two values after
the change.
 That product used to be $s_t \cdot s_{t'}$, and becomes $s_t\cdot
 s_{t'} -x(s_{t'}-s_t)-x^2$ --- the latter is smaller
 than $s_t \cdot s_{t'}$ since $x>0$ and $s_{t'} > s_t$.
Note also that at each step one of the $s_t$'s stops
being considered (either because it becomes equal to $\frac23$ or
equal to $0$) --- therefore the algorithm terminates. At termination
there will exist at most one $s_t$ with value different from $\frac23$
and $0$. Furthermore, recalling that $s = \sum_{t=1}^m s_t$, we conclude that
 then there will exist
exactly $\left\lceil \frac{s}{\nicefrac 23}\right\rceil$ different
$s_t$'s with value $\nicefrac23$, one with value $0 \le s - \left\lceil
\frac{s}{\nicefrac 23}\right\rceil \cdot \frac23 < \frac23$, and all the others having null value.

Given that $s \le \frac23 \log_5 \eta^{-1}$, we can then minimize the former probability with
\begin{align*}
\Pr[ \forall t, X_t \text{ equals the value } x_t
  \text{ that maximizes } \Pr[X_t = x_t]] &= \prod_{t=1}^m (1-s_t)\\
&\ge 3^{-\left\lceil \frac{s}{\nicefrac 23}\right\rceil-1}\\
& \ge  3^{-\left\lceil \log_5 \eta^{-1}\right\rceil-1}\\
& \ge \frac19 \cdot 3^{- \log_5 \eta^{-1}}\\ 
& \ge \frac19 \eta^{\frac1{\log_3 5}} \ge \frac{\eta}9.
\end{align*}
If these sequence of votes guarantees that the unknown urn $i$ loses
the election, we are done. Otherwise, we exchange the roles of urns
$i$ and $j$. 

Recall that $\frac56 \le \frac{p_i}{p_j},\frac{1-p_i}{1-p_j} \le \frac65$ ---
and therefore, for each $t$, $s'_t \le \frac65 s_t \le \frac45$.
Indeed, let $\{a,b,c\} = \{1,\nicefrac12,0\}$ be such that $p_{a,t}
\le p_{b,t} \le p_{c,t}$. If one lets
$p'_{1,t},p'_{\nicefrac12,t}, p'_{0,t}$ be the probabilities that
voter $t$, with unknown urn $j$, will, respectively, vote for $i$, for
an urn other than $i$ and $j$, and for urn $j$, then we have that
$p'_{1,t} \le \frac65 p_{1,t}, p'_{\nicefrac12,t} \le \frac65
p_{\nicefrac12,t}$ and $p'_{0,t} \le \frac65 p_{0,t}$. Therefore
$p'_{a,t} + p'_{b,t} \le \frac65 \left(p_{a,t}+p_{b,t}\right)$. It
follows that $s'_t = 1-p'_{c,t} = p'_{a,t} + p'_{b,t} \le \frac65
\left(p_{a,t}+p_{b,t}\right) = \frac65 s_t$, which is upper bounded by
$\frac65 \cdot s_t\le \frac65 \cdot \frac23 =
\frac45$.

\smallskip

Observe that the sum $s'$ of the $s'_t$'s, $s' = \sum_{t=1}^m s'_t$,
is then at most $\frac65$ times
the sum $s$ of the $s_t$'s; that is, $s' \le \frac65 s \le \frac4{5} \log_5 \eta^{-1}$.

\medskip

Let $X'_t$ be the random variable that, if the unknown urn is $j$, has
value $1$ if the $t$-th voter votes for urn $i$, $0$ if she
votes for urn $j$, and $\nicefrac12$ otherwise; we have:
$$\Pr\left[ \forall t, X'_t \text{ equals the value } x_t
  \text{ that maximizes } \Pr[X_t = x_t]\right] = \prod_{t=1}^m
(1-s'_t).$$
Using the same greedy algorithm as before, but moving mass $x =
\min\left(\frac45-s'_{t'}, s'_{t}\right)$ from couples of $s'_t$'s
such that $s'_{t'} < \frac45$ and $s'_t > 0$, $s'_{t'} \ne s'_t$, we get
that the previous product is minimized  when exactly $\left\lceil
\frac{s'}{\nicefrac45} \right\rceil$ distinct $s'_t$'s exist having
value $\nicefrac45$, one having value $0 \le s' - \left\lceil 
\frac{s}{\nicefrac 45}\right\rceil \cdot \frac45 < \frac45$, and the
rest having null value. Then,
\begin{align*}
\Pr\left[ \forall t, X'_t \text{ equals the value } x_t
  \text{ that maximizes } \Pr[X_t = x_t]\right] &= \prod_{t=1}^m
(1-s'_t)\\
&\ge 5^{-\left\lceil \frac{s'}{\nicefrac 45}\right\rceil-1}\\
&\ge 5^{-\left\lceil \frac4{5} \cdot \frac{\log_5\eta^{-1}}{\nicefrac 45}\right\rceil-1}\\
& \ge  5^{-\left\lceil \log_5 \eta^{-1}\right\rceil-1} \ge \frac{\eta}{25}.
\end{align*}
Now, if urn $i$ won with this sequence of votes, it follows that $j$
cannot win. 

\smallskip

We have shown that if $\var[X] \le \frac1{72} \log_5
\eta^{-1}$, then the probability of winning is at most
$1-\frac{\eta}{25}$. We now assume $\var[X] > \frac1{72} \log_5
\eta^{-1}$. We will use the following anti-concentration inequality
 (see Theorem
7.3.1 in \cite{mv}, and
\cite{f43}) to finish the proof:
\begin{theorem}[\cite{f43,mv}]\label{inverted_tail}
Let $X = \sum_{i=1}^n X_i$, where $X_i$ are independent random
variables, with $X_i \in [0,1]$, for $i = 1, \ldots, n$. Let $\sigma^2
= \var[X]$ be $\sigma^2 \ge 40000$. Then, for each $t \in \left[0,
  \frac{\sigma^2}{100}\right]$, it holds that
$$\Pr\left[X \ge E[X] + t\right] \ge c \cdot
\exp\left(-\frac{t^2}{3\sigma^2}\right),$$
for some universal constant $c > 0$.
\end{theorem}

We apply Theorem~\ref{inverted_tail} on the random variable $X= X(i,j)$, choosing $t=
\sqrt{\frac{64e^7m\var[X]}{n^3\epsilon^{-2}}}$, if $\Delta_i(j) \ge 0$, and $t = 0$ otherwise. This choice is valid since
$$0 \le \frac{t}{\var[X]} \le \sqrt{m\cdot\frac{64e^7 }{n^3 \epsilon^{-2} \var[X]}} < 
\sqrt{m\cdot\frac{4608 e^7  \ln 5 }{n^3 \epsilon^{-2} \ln \eta^{-1}}} \le \frac1{100},$$
where the latter holds if $m \le \frac{1}{46080000 e^7 \ln 5} \cdot  n^3
\epsilon^{-2} \ln \eta^{-1}$.

\smallskip

 We also need $\var[X] \ge 40000$ to
apply Theorem~\ref{inverted_tail}.
Since $\var[X] > \frac1{72} \log_5 \eta^{-1}$, and $\eta \le H$, we choose $H$ to be
$H = 5^{-2880000}$, obtaining $\var[X] > 40000$.

\medskip

Observe that $E[X] = \frac m2 - \frac m2 \cdot \Delta_{i}(j)$.
We show that the event ``$X \ge E[X] + t$'' implies the event ``$X \ge \frac m2$'' (which directly implies that the unknown urn $i$ will not win the election).

If $\Delta_i(j) < 0$, the claim is trivial, since then $E[X] > \frac m2$, and $t$ is non-negative. Otherwise, by the bound
$\var[X] \ge m \cdot \frac{\left|R_i - B_i\right|}{64}$, we get
$$t \ge  m \cdot e^{\nicefrac72} \epsilon \cdot
\sqrt{\frac{\left|R_i-B_i\right|}{n^3}} \ge \frac{m}2 \cdot \Delta_i(j),$$
which proves that $X \ge E[X] + t \implies X \ge \frac m2$.

Applying Theorem~\ref{inverted_tail}, we get
\begin{align*}
\Pr\left[X \ge E[X] + t\right] &\ge c \cdot
\exp\left(-\frac{t^2}{3\var[X]}\right) \\
&= c \cdot
\exp\left(-\frac{64e^7m}{3n^3\epsilon^{-2}}\right) \\
& \ge c \cdot \exp\left(-\frac{64}{3\cdot46080000\cdot \ln5}\cdot \ln \eta^{-1}\right)\\
& \ge c \cdot \eta.
\end{align*}
The proof is then complete.
\end{proof}

\section{An Upper Bound for Many Signals}

In this section we consider the voting problem in its full generality:
we have a set of $n \ge 2$ urns, with each urn $i = 1,\ldots,n$ inducing a distinct probability
distribution $P_i = (p_{i,1}, p_{i,2}, \ldots,
p_{i,C})$ over a set of $C$ signals or colors.
Let $\epsilon$ be the minimum $\ell_1$ distance between the
distributions $P_i$:
$$\epsilon = \min_{i \ne j} \ell_1(P_i,P_j) = \min_{i \ne j} \sum_{c = 1}^C
\left|p_{i,c} - p_{j,c}\right|.$$
Observe that when $C = 2$, this parameter $\epsilon$ is twice the
one that we used in Section~\ref{sec:twosig}.

\begin{theorem}\label{multicolor}
There exists a voting scheme that, using $m =
\Theta\left(\frac{(C \log C)^2 n^3}{\epsilon^2} \ln
\frac{n}{\eta}\right)$ voters, guarantee that the unknown urn wins
with probability at  least $1-\eta$.
\end{theorem}
The proof of this Theorem spans two subsections
(Sections~\ref{sec:fltwosig} and~\ref{sec:manysig}).
In Section~\ref{sec:fltwosig} we
generalize the bichromatic voting scheme so (a) to treat urns that are not
``well-separated'' as if they were the same -- this virtually increases the
separation parameter $\epsilon$ -- and (b) to guarantee, under some
conditions, that the equivalent of the $M$ parameter of the
bichromatic  voting scheme of Section~\ref{sec:twosig} is not just
upper bounded by $O(n^2\epsilon^{-1})$, but is actually asymptotic to $\Theta(n^2
\epsilon^{-1})$.

 In Section~\ref{sec:manysig}, we use both these
properties to devise a new voting scheme that uses the generalized
bichromatic one as a black box. The main idea of the multicolor voting
scheme is to force voters to view the urns as bichromatic ones: each
voter will choose a color $c$ at random, and consider each urn as a
bichromatic urn with colors $c, \bar{c}$ --- that is, she will imagine that there are only two colors: ``$c$'' and ``any color other than
$c$''. Using this trick directly with the bichromatic voting scheme of section~\ref{sec:twosig} would decrease to 0 the minimum
distance between urns in the worst case. We do not want the
separation between urns to decrease --- since that would increase the
minimum number of voters needed for the election to be successful ---  this is where property (a)
of the generalized bichromatic voting scheme becomes pivotal. Also, we
need a way to aggregate the votes given to each single
urn, in each of the $(c,\bar{c})$ bichromatic instances; this has to
be done in a way that guarantees that 
the right urn  will win
with high probability. We manage to do this by leveraging on property (b).

\subsection{A More Flexible Upper Bound with Two Signals}\label{sec:fltwosig}

To build a framework that can be used to handle the case of $C > 2$
signals,  it is useful to consider a more general formulation
of the bichromatic problem in which certain options can induce identical 
distributions over signals (and hence be indistinguishable from each other).
We present the analysis in the language of urns and colored balls.

Thus, suppose we have a collection of $n$ urns, labeled
$p_{i}$ for $i = 1, 2, \ldots, n$. With a slight abuse of notation we
let $p_{i}$ and $1-p_{i}$
be, respectively, the fraction
of blue balls, and of red balls, in urn $p_{i}$. We assume
w.l.o.g. that $0 \le p_1 \le p_2 \le \cdots \le p_n \le 1$.

\medskip

We assume that one urn is adversarially chosen as the correct one
(we will also refer to this as the {\em unknown} urn).
Then each player
draws a ball from the urn and votes for the name of an urn 
based on the color they observe.
For this general version with indistinguishable urns, we will be
interested in the probability that the urn receiving the most votes
has the same distribution as the correct one;
this general formulation is for the sake of the multi-color case later.

\medskip

We describe the strategy that the players will use to randomly
choose which vote to cast. First of all, for some $n' \ge 10$,
choose  $0 \le p'_1
< p'_2 < \cdots < p'_{n'} \le 1$. Let $\eps = \min_{1 \le i \le n' - 1} \left(p'_{i+1} -
p'_i\right)$. We require that (a) for $1 \le k \le
\left\lceil\frac{n'-1}3 \right\rceil = K$ it holds
that $p'_{k+1} - p'_k \le 2\eps$ and $p'_{n'-k + 1} - p'_{n'-k} \le
2\eps$, and (b) $p'_{K + 1} \le (2K+1)\epsilon$ and $p'_{n'-K} \ge 1 - (2K+1)\epsilon$.

\medskip

The $p'_i$'s are called the {\em landmarks} of the voting scheme.
\begin{itemize}
\item[1.] Let 
$\ds{b_k = \sum_{\ell=1}^{k-1} \frac{2-(p'_{\ell+1} + p'_{\ell})}{p'_{\ell+1} - p'_{\ell}}}$ and
$\ds{r_k = \sum_{\ell=k}^{n'-1} \frac{p'_{\ell+1} + p'_{\ell}}{p'_{\ell+1}
    - p'_{\ell}}}$ for $k =1, \ldots, n'$.
\item[2.] Let $\phi: \{p_1,\ldots, p_n\} \rightarrow
  \{1,\ldots,n'\}$ be a mapping from urns to landmarks' indices,
  defined so that $\phi(p_i) = k$ if $k$ 
  maximizes $p_i b_k + (1-p_i) r_k$ (ties can be broken arbitrarily).
\item[3.] 
Then define $R = \sum_{k=1}^{n'} \left(\left(\left|\phi^{-1}(k)\right| +
1\right) \cdot r_k\right)$ and $B = \sum_{k=1}^{n'}
  \left(\left(\left|\phi^{-1}(k)\right| + 1\right) \cdot b_k\right)$, and set $M = \max(R, B)$.
\item[4.] The probability that a voter will vote for $p_j$ if a blue ball is drawn is
$$\Pr_{\mathbf{P} \sim f(\text{blue})}[\mathbf{P} = p_{j}] = M^{-1} \cdot \left(b_{\phi(p_j)} + \frac{M -B}{n}\right) = B_j.$$
\item[5.] The probability that a voter will vote for $p_j$ if a red ball is drawn is
$$\Pr_{\mathbf{P} \sim f(\text{red})}[\mathbf{P} = p_{j}] =
  M^{-1} \cdot \left(r_{\phi(p_j)} + \frac{M -R}{n}\right) = R_j.$$
\end{itemize}

It is easy to check that the two probability  distributions $(B_1, B_2, \ldots,
B_n)$ and $(R_1,R_2,\ldots, R_n)$ are  well-defined (their values are non-negative and they both
sum up to one). Observe that $B_i = B_j$ and $R_i = R_j$ if $\phi(p_i) =
\phi(p_j)$.

\medskip

For a given urn $p_i$, let $k_i^+$
 be the smallest positive
index  such that
$p'_{k^+_i} \ge p_i$, if such an index exists, and $k^-_i$ be the largest index such that
$p'_{k^-_i} \le p_i$, again if the index exists; observe that at least one of $k_i^+$ and $k^-_i$
has to exist since $n' \ge 10$. We show the following lemma:

\begin{lemma}
For each $i = 1,\ldots, n$, $\phi(p_i)$ is either
equal to $k^+_i$ or to $k^-_i$. If, for some $i$, we have $p_i = p'_k$ it follows that $\phi(p_i) = k^+_i
= k^-_i = k$.
\end{lemma}
\begin{proof}
For an arbitrary $k$ it holds that
\begin{align*}
p_i \cdot b_k + (1-p_i) \cdot r_k &= p_i \cdot \sum_{\ell=1}^{k-1}
\frac{2-(p'_{\ell+1} + p'_{\ell})}{p'_{\ell+1} - p'_{\ell}} + (1-p_i)
\cdot \sum_{\ell=k}^{n'-1} \frac{p'_{\ell+1} + p'_{\ell}}{p'_{\ell+1}- p'_{\ell}}\\
&= \sum_{\ell=1}^{k-1}
\frac{2p_i}{p'_{\ell+1} - p'_{\ell}} + \sum_{\ell=k}^{n'-1}
\frac{p'_{\ell+1} + p'_{\ell}}{p'_{\ell+1}- p'_{\ell}} - p_i \cdot \sum_{\ell=1}^{n'-1}
\frac{p'_{\ell+1} + p'_{\ell}}{p'_{\ell+1}- p'_{\ell}}.
\end{align*}
The latter sum does not depend on $k$. Therefore $p_i b_k +
(1-p_i)r_k$ is maximized with a $k$ that maximizes the total of the former two
sums.

Suppose that $k^-_i$ exists and that $k < k^-_i$ maximizes the expression; then, by increasing $k$ to
$k+1 \le k^-_i$, we
would remove the term $\frac{p'_{k+1} + p'_k}{p'_{k+1} - p'_k}$ from
the second sum, and add the term $\frac{2p_i}{p'_{k+1}-p'_k}$ to the
first one. Since, by definition $p'_k < p'_{k+1} \le p'_{k^-} \le
p_i$, we have $p'_{k+1} + p'_k < 2p_i$, and therefore the total value
of the first two sums would increase. It follows that $k < k^-_i$ cannot
maximize the expression.

Analogously, suppose then that $k^+_i$ exists and that $k > k^+_i$ maximizes the expression; by
decreasing $k$ to $k-1 \ge k^+_i$, we remove the term
$\frac{2p_i}{p'_{k}-p'_{k-1}}$ from the first sum, and add the term
$\frac{p'_{k} + p'_{k-1}}{p'_{k} - p'_{k-1}}$ to the second one. Since
$p_i \le p'_{k^+} \le p'_{k-1} < p'_k$, we obtain that $2p_i < p'_{k-1} +
p'_k$ and therefore we increase the total value of the first two
sums; thus $k > k^+_i$ does not maximize the expression.
\end{proof}

We now turn to computing the
probability that a player will vote for $p_j$  given that the correct,
adversarially chosen, distribution is $p_i$:
$$\Pr_{\substack{\mathbf{X} \sim p_i\\ \mathbf{P} \sim f(\mathbf{X})}}\left[\mathbf{P} =
    p_j\right]  =  p_i \Pr_{\mathbf{P} \sim
  f(\text{blue})}[\mathbf{P} =  p_j] + (1-p_i) 
\Pr_{\mathbf{P} \sim f(\text{red})}[\mathbf{P} =  p_j] = p_i B_{j} + (1-p_i) R_{j} = E_{i}(j).$$

 We compute
the difference between the probabilities that a vote for urn $p_i$
and a vote for urn $p_j$ are cast, given that the correct urn is
$p_i$:
$$\Delta_{i}(j) = E_{i}(i) - E_{i}(j).$$

We will lower-bound $\Delta_{i}(j)$ to bound the number of voters needed
to let the voting scheme be successful with high probability. 

\begin{lemma}
For each $1\le i,j \le n$, it holds that
$$\Delta_{i}(j) \ge \left\{\begin{array}{cr}
\frac{\left|\phi(p_i) - \phi(p_j)\right|}M & \text{if } p_i = p'_{k^+_i} =
p'_{k^-_i}\\
\frac{\max(\left|\phi(p_i) - \phi(p_j)\right| - 1,0)}M & \text{otherwise}
\end{array}\right.$$
\end{lemma}
\begin{proof}
We make the expression of $\Delta_i(j)$ explicit:
$$\Delta_i(j) = p_{i} \left(B_i
-B_j\right)
 + (1-p_{i})  \left(R_i - R_j\right) = p_i \cdot \frac{b_{\phi(p_i)} -
   b_{\phi(p_j)}}{M} + (1-p_i) \cdot \frac{r_{\phi(p_i)} -
   r_{\phi(p_j)}}M.$$

Suppose
first that $\phi(p_j) < k^-_i$; then
\begin{align*} 
M \cdot \Delta_{i}(j) & = p_{i} \cdot \sum_{\ell=\phi(p_j)}^{\phi(p_i)-1}
\frac{2-(p'_{\ell+1} + p'_{\ell})}{p'_{\ell+1}-p'_{\ell}} - (1-p_i) \cdot
\sum_{\ell=\phi(p_j)}^{\phi(p_i)-1} \frac{p'_{\ell+1}+p'_{\ell}}{p'_{\ell+1}-p'_{\ell}} 
= \sum_{\ell=\phi(p_j)}^{\phi(p_i)-1} \frac{2 \cdot p_i - (p'_{\ell+1}+p'_{\ell})}{p'_{\ell+1}-p'_{\ell}}\\
&\ge \sum_{\ell=\phi(p_j)}^{k^-_i-1} \frac{2 \cdot p_i -
  (p'_{\ell+1}+p'_{\ell})}{p'_{\ell+1}-p'_{\ell}}\ge
\sum_{\ell=\phi(p_j)}^{k^-_i-1} \frac{2 \cdot p'_{\ell+1} -
(p'_{\ell+1}+p'_{\ell})}{p'_{\ell+1}-p'_{\ell}} =
\sum_{\ell=\phi(p_j)}^{k^-_i-1} \frac{p'_{\ell+1} -
p'_{\ell}}{p'_{\ell+1}-p'_{\ell}} = \left|k^-_i-\phi(p_j)\right|,
\end{align*}
where the first inequality follows from $\phi(p_i) \ge k^-_i$ and the second
from $p_i \ge p'_{k^-_i} \ge p'_{\ell+1}$.

If, instead, $\phi(p_j) > k^+_i$ we have:
\begin{align*} 
M \cdot \Delta_{i}(j) & = -p_{i} \cdot \sum_{\ell=\phi(p_i)}^{\phi(p_j)-1}
\frac{2-(p'_{\ell+1} + p'_{\ell})}{p'_{\ell+1}-p'_{\ell}} + (1-p_i) \cdot
\sum_{\ell=\phi(p_i)}^{\phi(p_j)-1} \frac{p'_{\ell+1}+p'_{\ell}}{p'_{\ell+1}-p'_{\ell}} 
= \sum_{\ell=\phi(p_i)}^{\phi(p_j)-1} \frac{(p'_{\ell+1}+p'_{\ell}) - 2 p_i}{p'_{\ell+1}-p'_{\ell}}\\
&\ge \sum_{\ell=k^+_i}^{\phi(p_j)-1} \frac{
  (p'_{\ell+1}+p'_{\ell}) -2p_i}{p'_{\ell+1}-p'_{\ell}}\ge
\sum_{\ell=k^+_i}^{\phi(p_j)-1} \frac{
(p'_{\ell+1}+p'_{\ell})-2p'_{\ell}}{p'_{\ell+1}-p'_{\ell}} =
\sum_{\ell=k^+_i}^{\phi(p_j)-1} \frac{p'_{\ell+1} -
p'_{\ell}}{p'_{\ell+1}-p'_{\ell}} = \left|k^+_i-\phi(p_j)\right|,
\end{align*}
where the first inequality follows from $\phi(p_i) \le k^+_i$ and the second
from $p_i \le p'_{k^+_i} \le p'_{\ell}$.

\medskip

Also, recall that $\phi(p_i)$ is an index that maximizes $p_i b_{\phi(p_i)} +
  (1-p_i) r_{\phi(p_i)}$. Since
$$\Delta_i(j)  =
 M^{-1} \cdot \left(
  \left(p_i b_{\phi(p_i)} + (1-p_i) r_{\phi(p_i)}\right) - \left(p_i
  b_{\phi(p_j)} + (1-p_i) r_{\phi(p_j)}\right)\right),$$
 we have that $\Delta_{i}(j) \ge 0$ for each
ordered couple of urns $i,j$. The statement follows.
\end{proof}

\medskip

We now give an upper bound on the probability that the correct urn
will be chosen by a voter. 
Note, somewhat counter-intuitively, that the probability of a correct vote
is higher when this upper bound is smaller --- this is because the
$\Delta_{i}(j)$ are additive gaps, not
multiplicative ones, and so by making the upper bound on the expected
number of votes for the correct urn smaller, the gap $\Delta_{i}(j)$
becomes larger relative to the mean.

\begin{lemma}
It holds that
$$E_i(i)  
\le \frac{2(n'-1)}{\epsilon M} + \frac1{n}.$$
\end{lemma}
\begin{proof}
Recall that the correct urn is $p_i$. We upper-bound the
probability that a vote will actually go to $p_i$:
\begin{align*}
E_i(i)  & = p_i \cdot M^{-1} \cdot \left(b_{\phi(p_i)} + \frac{M -B}{n}\right) + (1 -
p_i) \cdot M^{-1} \cdot \left(r_{\phi(p_i)} + \frac{M -R}{n}\right)\\
& \le p_i \cdot \left(M^{-1} \cdot b_{\phi(p_i)} + \frac{1}{n}\right) + (1 -
p_i) \cdot \left(M^{-1} \cdot r_{\phi(p_i)} + \frac{1}{n}\right).
\end{align*}
Observe that, by the definition of $\epsilon$, we have that
$\ds{b_i = \sum_{\ell=1}^{i-1} \frac{2-(p'_{\ell+1} +
  p'_{\ell})}{p'_{\ell+1} - p'_{\ell}}}$ is at most $\ds{b_i \le
\frac{2(i-1)}{\epsilon}}$, and that $\ds{r_i = \sum_{\ell=i}^{n'-1}
\frac{p'_{\ell+1} + p'_{\ell}}{p'_{\ell+1} - p'_{\ell}} \le
\frac{2(n'-i)}{\epsilon}}$. Thus,
$$
E_i(i)   \le p_i \cdot \left(\frac{2(i-1)}{\epsilon M} + \frac{1}{n}\right) + (1 -
p_i) \cdot \left(\frac{2 (n'-i)}{\epsilon M} + \frac{1}{n}\right)
\le \frac{2(n'-1)}{\epsilon M} + \frac1{n}.
$$
\end{proof}

We now give an upper bound on $M$. This will allow
us to apply a Chernoff bound and prove the main theorem.

\begin{lemma}\label{Mbound_lemma}
It holds that 
$$ \frac1{81} \cdot \frac{(n'-1)
  \cdot (n+n')}{\epsilon} \le M \le 2 \cdot \frac{(n'-1) \cdot (n+n')}{\epsilon}.$$
\end{lemma}
\begin{proof}
Recall that $M = \max(R,B)$; we will upper bound $R + B$ to get an
upper bound on $M$:
\begin{align*}
R + B &= \sum_{k=1}^{n'}
\left(\left(\left|\phi^{-1}\left(p'_k\right)\right| + 1\right)
\cdot(r_k + b_k)\right) \le \sum_{k=1}^{n'}
(\left(\left|\phi^{-1}\left(p'_k\right)\right|+1\right)
\cdot(r_1 + b_{n'})) \\
& = \sum_{k=1}^{n'} \left(\left(\left|\phi^{-1}\left(p'_k\right)\right|+1\right)\cdot
\sum_{\ell=1}^{n'-1} \frac{2}{p_{\ell+1}-p_{\ell}}\right) \\
\le& \sum_{k=1}^{n'} \left(\left(\left|\phi^{-1}\left(p'_k\right)\right|+1\right) \cdot (n'-1) \cdot \frac{2}{\epsilon}\right) =
\frac{2\cdot (n'-1) \cdot (n+n')}{\epsilon}.
\end{align*}
It follows that
\begin{equation}
M \le \frac{2\cdot (n'-1) \cdot (n+n')}{\epsilon}.
\end{equation}

We now move on to the lower bound.
Recall that $n' \ge 10$, and that, for each $k =
1,2,\ldots,\left\lceil\frac{n'-1}3\right\rceil = K$ (resp., $k = n' - K,
\ldots, n' - 1$) we
have $p'_{k+1} - p'_k \le 2\epsilon$; also,
$p'_{K+1}\le (2K+1)\epsilon$ and $p'_{n'-K} \ge 1-(2K+1)\epsilon$.
 Observe that
$$2K+1=  2\left\lceil\frac{n'-1}3\right\rceil +1 \le 2
\cdot \frac{n' + 1}3 + 1 \le 2
\cdot \frac{n' + \frac52}3  \le 2
\cdot \frac{n' + \frac{n'}4}3 \le \frac{5}{6} \cdot n'.$$
Since $\epsilon = \min_{1 \le i \le n' - 1} p'_{i+1} - p'_i$, we have $\epsilon \le \frac1{n' - 1}$, and
$$(2K+1)\epsilon \le \frac56 \cdot \frac{n'}{n'-1} \le \frac{25}{27}.$$

We are now ready to lower bound $M$:
\begin{eqnarray*}
M &=& \max(R, B) \ge \frac{R+B}2 = \frac12\cdot\sum_{k=1}^{n'} \left(\left(\left|\phi^{-1}(k)\right| +
1\right) \cdot \left(r_k + b_k\right)\right) \\
&=& \frac12 \cdot \sum_{k=1}^{n'} \left(\left(\left|\phi^{-1}(k)\right| +
1\right) \cdot \sum_{\ell=k}^{n'-1} \frac{p'_{\ell+1} + p'_{\ell}}{p'_{\ell+1}
    - p'_{\ell}}\right) + \frac12 \cdot \sum_{k=1}^{n'}
\left(\left(\left|\phi^{-1}(k)\right| +1\right) \cdot \sum_{\ell=1}^{k-1}\frac{2-(p'_{\ell+1} + p'_{\ell})}{p'_{\ell+1}
    - p'_{\ell}} \right)\\
&\ge& \frac12 \cdot \sum_{k=1}^{n'} \left(\left(\left|\phi^{-1}(k)\right| +
1\right) \cdot \sum_{\ell=\max(k, n'-K)}^{n'-1} \frac{p'_{\ell+1} + p'_{\ell}}{p'_{\ell+1}
    - p'_{\ell}}\right) +\\
&& \frac12 \cdot \sum_{k=1}^{n'}
\left(\left(\left|\phi^{-1}(k)\right| +1\right) \cdot \sum_{\ell=1}^{\min(k-1,K)}\frac{2-(p'_{\ell+1} + p'_{\ell})}{p'_{\ell+1}
    - p'_{\ell}} \right)\\
&\ge& \frac12 \cdot \sum_{k=1}^{n'}
\left(\left(\left|\phi^{-1}(k)\right| +
1\right) \cdot \sum_{\ell=\max(k, n'-K)}^{n'-1}
\frac{2(1-(2K+1)\epsilon)}{2\epsilon}\right) +\\
&& \frac12 \cdot \sum_{k=1}^{n'}
\left(\left(\left|\phi^{-1}(k)\right| +1\right) \cdot
\sum_{\ell=1}^{\min(k-1,K)}\frac{2-2\cdot(2K+1)\epsilon}{2\epsilon}
\right),
\end{eqnarray*}
using the upper bound we previously obtained for $(2K+1)\epsilon$, we obtain:
\begin{eqnarray*}
M&\ge& \frac12 \cdot \sum_{k=1}^{n'} \left(\left(\left|\phi^{-1}(k)\right| +
1\right) \cdot \sum_{\ell=\max(k, n'-K)}^{n'-1} \frac{2}{27\epsilon}\right) + \frac12 \cdot \sum_{k=1}^{n'}
\left(\left(\left|\phi^{-1}(k)\right| +1\right) \cdot
\sum_{\ell=1}^{\min(k-1,K)}\frac{2}{27\epsilon} \right)\\
&\ge& \frac1{27\epsilon} \cdot \sum_{k=1}^{n'} \left(\left(\left|\phi^{-1}(k)\right| +
1\right) \cdot \left(n'-\max(k,n'-K) + \min(k-1,K)\right)\right)\\
&\ge& \frac1{27\epsilon} \cdot \sum_{k=1}^{n'} \left(\left(\left|\phi^{-1}(k)\right| +
1\right) \cdot \left(\min(n'-k,K) + \min(k-1,K)\right)\right),
\end{eqnarray*}
observe that $(n'-k)+(k-1) = n'-1$ and therefore at least one of
$(n'-k)$ and $k-1$ is at least $$\frac{n'-1}2 \ge  \frac{\frac32 n' -
  \frac32}3 = \frac{n' + \frac{n'-3}2}3 \ge \frac{n' + \nicefrac72}3 >
\left\lceil\frac{n'-1}3\right\rceil = K.$$
Therefore, at least one of $\min(n'-k,K), \min(k-1,K)$ is at least
$K$. Then,
$$M \ge \frac1{27\epsilon} \cdot \sum_{k=1}^{n'} \left(\left(\left|\phi^{-1}(k)\right| +
1\right) \cdot K\right) = \frac K{27\epsilon} \cdot \sum_{k=1}^{n'} \left(\left|\phi^{-1}(k)\right| +
1\right) = \frac K{27\epsilon} \cdot \left(n + n'\right) \ge \frac{(n'-1)
  \cdot (n+n')}{81\cdot\epsilon}.$$
\end{proof}

\comment{
$p_{\left\lceil\frac n2\right\rceil} \le \nicefrac12$
($p_{\left\lfloor\frac n2\right\rfloor} \ge \nicefrac12$). Then, at
least
$\nicefrac n8$ of those $p_i$'s will be such that $p_{i+1} -
p_{i} \le \frac2n$ --- for otherwise they would span more than just the
interval $\left[0,\nicefrac12\right]$. Then, $B = \sum_{k=0}^n (N_k
\cdot b_k) \ge $}

Therefore, going back to the probability that an urn identical to 
the correct urn is voted for, we have
$$
E_i(i)  \le  \frac{2(n'-1)}{\epsilon M} + \frac1{n} \le
\frac{2(n'-1)}{\epsilon\cdot \frac1{81} \cdot
  \frac{(n'-1)(n+n')}{\epsilon}} = \frac{162}{n+n'}
$$

\subsection{An Upper Bound for Many Signals}\label{sec:manysig}

In this section we consider the voting problem in its full generality:
we have a set of $n \ge 2$ urns, with each urn $i = 1,\ldots,n$ inducing a distinct probability
distribution $P_i = (p_{i,1}, p_{i,2}, \ldots,
p_{i,C})$ over a set of $C$ colors. 
Let $\epsilon$ be the minimum $\ell_1$ distance between the
distributions $P_i$:
$$\epsilon = \min_{i \ne j} \ell_1(P_i,P_j) = \min_{i \ne j} \sum_{c = 1}^C
\left|p_{i,c} - p_{j,c}\right|.$$

It turns out that the bichromatic scheme from Section~\ref{sec:fltwosig}
has already laid much of the groundwork for the multi-color case.
Each voter $u$ will behave as follows:\begin{itemize}
\item[1.] First, $u$ will choose a color $c = c(u)$ uniformly 
at random from among all the colors.
Voter $u$ will then imagine the urns as inducing a bichromatic instance
by imagining all colors other than $c$ as a single color $\bar{c}$.
In this way, urn $i$ becomes a bichromatic urn with distribution
$(p_{i,c}, 1 -p_{i,c})$ over its two colors.

\item[2.] Then, voter $u$ will choose an integer $t = t(u) \in
  \left\{0,1,\ldots, T\right\}$ with $T= \left\lceil \log_3 C
  \right\rceil + 1$, in such a way that $\Pr[t = i] = \alpha^{-1}
  \cdot 3^{-T+i}$, where $\alpha = \sum_{i=0}^T 3^{-i}$.

 Observe that $\alpha < \sum_{i=0}^{\infty} 3^{-i} = \frac32$.

\item[3.] 
Voter $u$ will then apply the
bichromatic voting scheme from Section~\ref{sec:fltwosig}
to choose which urn to vote 
for. She will set $\{p_1, p_2, \ldots, p_n\} = \{p_{1,c}, p_{2,c},
\ldots, p_{n,c}\}$, for $i = 1, \ldots, n$; the sequence of the
$p'_i$'s will be defined as follows: \begin{itemize}
\item first she will pick a subsequence according to the following
  {\em marking algorithm}: set $w_1 = 0$ and mark all the $p_j$'s such
  that $p_j \le 3^{-t} \cdot \epsilon$; if some unmarked $p_j$
  remains, let $i = 2$, and
\item set $w_{i}$ to be the smallest unmarked $p_j$,
\item mark all the $p_j$'s
for which $\left|p_j - w_i\right| < 3^{-t} \cdot \epsilon$;
\item if some unmarked $p_j$ remains, repeat; otherwise, if $w_i \ne
  1$, set $w_{i+1} = 1$; then, stop.
\end{itemize}
\item[4.] let  $i^*$ be the length of the sequence $\{w_i\}$; if $i^*
  < 10$, the voter will add $10 - i^*$ elements to $\{w_i\}$: let $i$
  be such that $w_{i+1} - w_i$ is maximized; the voter will insert the values
  the values $w_i + \frac1{9} (w_{i+1} - w_i), w_i + \frac2{9}
  (w_{i+1} - w_i), \ldots, w_i + \frac{8}9 (w_{i+1}-w_i)$ in the list,
  keeping it sorted.

 The size of the sequence $\{w_i\}$ will then be at least 10.

 The voter
  will then define the sequences $x_{i,1}, x_{i,2}$ as $x_{i,1} = \frac{2w_{i} + w_{i+1}}3$,
  $x_{i,2} = \frac{w_i + 2w_{i+1}}3$, for $i = 1, \ldots, i^*-1$;
\item[5.] the voter then merges the sequences $w_i,
  x_{i,1}, x_{i,2}$, and sorts the resulting sequence increasingly; let
  $y_1 < y_2 < \cdots < y_{3i^*-2}$ be this sequence, and
  $\epsilon_{c,t}$ be its separation parameter: $\epsilon_{c,t} =
  \min_{i=1,\ldots, 3i^*-3} y_{i+1} - y_i$.
\item[6.] then the voter adds elements to $\{y_i\}$ in such a way that:\begin{itemize}
 \item[(a)] the minimum separation between adjacent elements remains at
  least $\epsilon_{c,t}$, \item[(b)] if the list has length $n'$, then for
  each $i = 1,\ldots, \left\lceil \frac{n'}3\right\rceil$, it holds
  that $y_{i+1} -y_i \le 2\epsilon_{c,t}$, and \item[(c)] 
if the list has length $n'$, then for
  each $i = \left\lfloor\frac{2n'}3\right\rfloor,\ldots,n'-1$, it holds
  that $y_{i+1} -y_i \le 2\epsilon_{c,t}$.\end{itemize} To do so, she applies the following
  algorithm:\begin{itemize}
\item[6.1] if (b) is not satisfied, i.e., if the list $\{y_i\}$ has
  currently length $n'$ and $i$ is a minimal index $i$ for which there exists elements
  $y_i, y_{i+1}$ such that $i \le \left\lceil\frac{n'}3\right\rceil$
  and $y_{i+1} - y_i > 2\epsilon_{c,t}$, then insert a new element between
  $y_i$ and $y_{i+1}$ of value $y_i +\epsilon_{c,t}$; this will increase the
  length of the list; repeat this step as
  long as (b) is not satisfied;
\item[6.2] if (c) is not satisfied, i.e., if the list $\{y_i\}$ has
  currently length $n'$ and there is a maximal index $i$ for which there
  exists elements
  $y_i, y_{i+1}$ such that $i \ge \left\lfloor\frac{2n'}3\right\rfloor$
  and $y_{i+1} - y_i > 2\epsilon_{c,t}$, then insert a new element between
  $y_i$ and $y_{i+1}$ of value $y_{i+1} -\epsilon_{c,t}$; repeat this step as
  long as (c) is not satisfied.
\end{itemize}
It is easy to prove that the above algorithm guarantees properties (a),
(b) and (c). Let $n'_{c,t}$ be the length of the final sequence $\{y_i\}$, $K_{c,t} =
\left\lceil\frac{n'_{c,t}}3\right\rceil$, and observe that
the algorithm also guarantees that (d) $n'_{c,t} \ge 10$, (e) $n'_{c,t} \le 3i^*-2 +2(K_{c,t}+1) \le 9i^*
+ 2 \le 9n + 2$ and (f) $y_{K_{c,t} + 1} \le (2K_{c,t}+1) \epsilon_{c,t}$
and $y_{n'_{c,t} - K_{c,t}} \ge
1-(2K_{c,t}+1)\epsilon_{c,t}$.
\end{itemize}

The just-defined bichromatic instance depends only on the original
multi-colored instance, on $c$ and on $t$ --- we use $(c,t)$-instance
to refer to the bichromatic instance induced by $c$ and $t$.

Observe that the separation parameter $\epsilon_{c,t}$ of the
$(c,t)$-instance will be at least $\epsilon_{c,t} \ge 3^{-t-3} \cdot
\epsilon$, since the $w_i$'s are at distance of at least
$3^{-t-2}\cdot\epsilon$ from each other\footnote{Before step 4 they
  were at distance at least $3^{-t}\epsilon$ from each other, and
  step 4 could have added new $w_i$'s at distance at least $3^{-t-2}$
  from each other.}, and  $x_{i,1}$, $x_{i,2}$ split
the interval between $w_i$ and $w_{i+1}$ in three equal parts --- the
subsequently added $y_i$'s do not induce gaps smaller that
$\epsilon_{c,t}$. Furthermore the number of landmarks of the
$(c,t)$-instance will be $10 \le n'_{c,t} \le 9n+2 \le 10 n$, since $n
\ge 2$. 

We are now ready to prove Theorem~\ref{multicolor}, using the machinery built in
Section~\ref{sec:fltwosig}.

\begin{proof}[Proof of Theorem~\ref{multicolor}]
First, given two urns $P_i,P_j$, we say that a color $c$ is {\em useful
  for $i,j$} if $\left|p_{i,c} - p_{j,c}\right| >
\frac{\epsilon}{3C}$. Observe that if $C_{i,j}$ is the set of useful
colors for urns $P_i,P_j$, we have
$$\sum_{c \in C_{i,j}} \left|p_{i,c} - p_{j,c}\right| >
\frac{2}3\cdot \epsilon.$$
Indeed, since there are only $C$ colors, the contribution to the
$\ell_1$ distance between $P_i$ and $P_j$ of their non-useful colors is less than $C \cdot
\frac{\epsilon}{3C} = \frac{\epsilon}3$. Given that the total distance
is at least $\epsilon$, it follows that the useful colors contribute
by more than $\frac{2\epsilon}3$ to the $\ell_1$ distance of $P_i,P_j$.

\medskip

Suppose that $i$ is the unknown urn.  Let $p_i = p_{i,c}$ and $p_j =
p_{j,c}$ in the $(c,t)$-bichromatic instance, for some $c,t$. Let $E_i^{(c,t)}(j)$ be
the expected number of votes that a voter will give to urn $j$, if $i$ is the unknown urn, in
the $(c,t)$-bichromatic instance.
The analysis of the
bichromatic case, guarantees that
$$E^{(c,t)}_i(j) \le E^{(c,t)}_i(i) \le \frac{162}{n+n'_{c,t}} \le \frac{162}n.$$
and that
 the difference $\Delta^{(c,t)}_i(j)
= E_i^{(c,t)}(i) - E_i^{(c,t)}(j)$ is at least
$$\Delta^{(c,t)}_i(j) \ge \frac{\max\left(\left|\phi(p_i) -
  \phi(p_j)\right|-1,0\right)}{M_{c,t}}.$$

\medskip

Fix a color $c \in C_{i,j}$ and let $t_{c,i,j}$ be
the smallest non-negative integer such that
$$\left|p_{i,c} - p_{j,c}\right| \ge \epsilon \cdot 3^{-t_{c,i,j}}.$$
Since $c$ is a useful color we have $\left|p_{i,c} - p_{j,c}\right| >
\frac{\epsilon}{3C}$, and therefore $0 \le t_{c,i,j} \le \left\lceil\log_3
C\right\rceil + 1 = T$.
By
$\epsilon_{c,t} \ge 3^{-t-3} \cdot \epsilon$, we obtain
$$\epsilon_{c,t_{c,i,j}} > \frac1{81} \left|p_{i,c} - p_{j,c}\right|.$$

Since $\left|p_i - p_j\right| \ge \epsilon \cdot 3^{-t_{c,i,j}}$, the
marking algorithm run by the voters 
will mark $p_i$ and $p_j$ at different
iterations --- therefore, there are at least three landmarks
between $p_i$ and $p_j$. It follows that $\left|\phi(p_i) -
\phi(p_j)\right| \ge 2$. Therefore,
$$\Delta^{(c,t_{c,i,j})}_i(j) \ge \frac{1}{M_{c,t_{c,i,j}}}\ge\frac1{2 \cdot
\frac{(n'_{c,t_{c,i,j}}-1)\cdot(n + n'_{c,t_{c,i,j}})}{\epsilon_{c,t_{c,i,j}}}} \ge
\frac{\epsilon_{c,t_{c,i,j}}}{22 n (n'_{c,t_{c,i,j}}-1)} \ge
\frac{|p_{i,c}-p_{j,c}|}{17820 \cdot n^2}.$$

Then,
\begin{align*}
\Delta_i(j) &= \sum_{c =1}^ C \sum_{t=1}^T \left(\frac1C
\cdot \frac1T \cdot \Delta_i^{(c,t)}(j)\right) \ge \frac1{C \cdot T} \cdot  \sum_{c \in
  C_{i,j}} \max_{t=1,\ldots,T} \Delta_i^{(c,t)}(j) \ge \frac1{C \cdot
T} \cdot \sum_{c \in
  C_{i,j}} \Delta_i^{(c,t_{c,i,j})}(j) \\
&\ge \frac1{C \cdot
T} \cdot \sum_{c \in
  C_{i,j}}
\frac{|p_{i,c}-p_{j,c}|}{17820 \cdot n^2} \ge \frac{\epsilon}{26730
  \cdot C \cdot T
  \cdot n^2} = x.
\end{align*}
and
$$E_i(j) = \sum_{c =1}^C \sum_{t=1}^T \left(\frac1C
\cdot \frac1T \cdot E_i^{(c,t)}(j)\right) \le \frac{162}{n}$$

\medskip

We choose $m = \left\lceil 7 \cdot 10^{12} \cdot \frac{ C^2 \cdot T^2
  \cdot n^3}{\epsilon^2} \ln \frac{n}{\eta}\right\rceil =
\Theta\left(\frac{(C \log C)^2 \cdot n^3}{\epsilon^2} \ln \frac{n}{\eta}\right)$ as the number
of voters, and we apply Chernoff bound (see
Theorem~\ref{thm:chernoff}), on $V_j$: the number of votes to urn $j$
in the election, if $i$ is the unknown urn. Observe that $E[V_j] = m
\cdot E_i(j)$, and furthermore:
\begin{align*}
\Pr\left[\left|V_j - E[V_j]\right| > \frac{x}{3} \cdot m\right]
  &= 
\Pr\left[\left|V_j - E[V_j]\right| > \frac{x}{3E_i(j)} \cdot
  E_i(j) \cdot m\right] \\
& \le 2
\exp\left(-\frac{x^2}{27 E^2_i(j)} \cdot E_i(j) \cdot
m\right)\\
 & \le 2
\exp\left(-\frac{x^2}{27 E_i(j)} \cdot m\right)\\
 & \le 2
\exp\left(-\frac{\epsilon^2}{27 \cdot 26730^2 \cdot C^2 \cdot T^2
  \cdot n^4 \cdot \frac{162}n} \cdot m\right)\\
 & \le 2
\exp\left(-\frac{\epsilon^2}{27 \cdot 26730^2 \cdot C^2 \cdot T^2
  \cdot n^4 \cdot \frac{162}n} \cdot m\right) \le \frac{\eta}{n}.
\end{align*}

Applying the Union Bound over all the urns, we have that each single urn $j$
will deviate by at most $\frac x3 m$ from its
expected number of votes with probability at least  $1-\eta$, and since the expected difference of the
number of votes of urn $i$ and $j$ if $i$ is the unknown urn is at
least $m \cdot \Delta_i(j) \ge m\cdot x$, we have that urn $i$ will
win the election with probability at least $1-\eta$.
\end{proof}

\section{Other Voting Systems}

In this section we study other important voting systems, assuming that
there are two types of signals; that is, assuming bichromatic urns.

\subsection{Cumulative Voting}

We show that {\em cumulative voting} requires a smaller  number of
voters for the election to succeed with high probability. In fact,
{\em cumulative voting} can be exploited to work with a number of voters as small as the number of samples used by the optimal centralized algorithm
(that is, the algorithm that, after sampling the minimum number of
balls, produces the right guess with high probability).

Like plurality voting, in the cumulative voting election system, each voter has a single vote to
cast; unlike plurality voting, though, the voter can split her vote arbitrarily between the
candidates:
\begin{definition}[Cumulative Voting]
Each voter assigns a score to each candidate, in such a way that no
score is negative and the sum of the scores assigned by a voter is
$1$. The total score of a candidate is the sum of the scores assigned
to that candidate by the voters. If there exists a candidate $i$
having a total score larger than the total score of each other
candidate $j \ne i$, then $i$ is the winner of the election.
\end{definition}

Given urns $p_1,p_2,\ldots,p_n$, the voting scheme we propose for cumulative voting is directly derived
from the plurality voting scheme we proposed earlier; in the new scheme,
there are only two possible votes: if a voter picks
a red (resp., blue) ball then she will vote $(R_1,R_2,\ldots, R_n)$ (resp., $(B_1,B_2,\ldots,B_n)$) --- that is, she will assign a weight of $R_i$ (resp., $B_i$) to candidate $P_i$, for $i=1,2,\ldots, n$. The $R_i$'s and
the $B_i$'s are those that we defined for the plurality voting scheme.

\begin{theorem}\label{thm:upperb_cv}
Let urns $p_1, p_2, \ldots, p_n$ be given, with urn $p_i$ having a
$p_i$ fraction of blue balls, and a $1-p_i$ fraction of red balls. Let $p_1 < p_2 < \cdots < p_n$.
Also, let $\epsilon$ be $\epsilon = \min_{1 \le i \le n-1} (p_{i+1} -
p_i)$. Then, for Cumulative Voting, $O\left(\epsilon^{-2} \ln \frac 1{\eta}\right)$ voters are
sufficient to guarantee a probability of at least $1 -\eta$ that
the correct urn wins the election.
\end{theorem}
\begin{proof}
Choose some $\eta \in (0,1)$, and suppose the number of players is
$$m =  \left\lceil 150 \cdot 
\epsilon^{-2} \cdot \ln \frac{2}{\eta}\right\rceil.$$

We start by using the Chernoff bound to show that the number of voters that pick a ball of some color is concentrated. If $p_i \ge \frac12$, let $X$ be the number of voters picking a blue ball; otherwise let $X$ be the number of voters picking a red ball.
In both cases, $X$ is the sum of iid binary random variables $X_j$,
with $X_j = 1$ with probability $\max(p_i,1-p_i)$ and $X_j = 0$ with
probability $\min(p_i,1-p_i)$. Then, $E[X] = m \cdot \max(p_i,1-p_i)$
and $\frac m2 \le E[X] \le m$. By Chernoff bound,
\begin{align*}
\Pr\left[\left|X - E[X]\right| > \frac{\epsilon }5 \cdot m\right]
& \le\Pr\left[\left|X - E[X]\right| > \frac{\epsilon}5 \cdot
  E[X]\right] 
 \le 2 \cdot \exp\left(-\frac{\epsilon^2}{75} \cdot
E[X]\right) \\
& \le 2 \cdot \exp\left(-\frac{\epsilon^2}{75} \cdot \frac m2 \right) \le
2 \cdot \exp\left(\ln\frac 2{\eta}\right) =  \eta,
\end{align*}
that is, with probability $1-\eta$ the absolute difference between the
number $m_b$ of blue (resp.,
the number $m_r$ of red) balls picked and the expectation $p_i \cdot m$ ($(1-p_i)
\cdot m$) by at most $\frac{\epsilon}5 m$.

For any $i, j \in [n]$, $i \ne j$, let $V_i(j)$ be the fractional number of votes to $P_j$ in
the random election, with unknown urn $P_i$. Let $D_i(j) = V_i(i) -
V_i(j)$.  Observe that urn $i$ beats urn $j$ in the election (with
unknown urn $i$) iff $D_i(j) > 0$. The random variable $D_i(j)$ is the sum of $m$ iid
random variables $D'_i(j)$ each taking value $B_i - B_j$ if the
corresponding voter picked a blue ball and $R_i - R_j$ if she picked a
red ball; we now bound the span $S$ of values of $D'_i(j)$ --- that is, we
bound $S = |(B_i - B_j) - (R_i - R_j)|$. Observe
that 
 $B_i - B_j = \frac{b_i - b_j}M$ and $R_i - R_j =
\frac{r_i - r_j}M$. Also,
$$b_i-b_j = \sum_{\ell =
0}^{i-1} \frac{2-(p_{\ell+1}+p_{\ell})}{p_{\ell+1}-p_{\ell}} -
\sum_{\ell =
  0}^{j-1}
\frac{2-(p_{\ell+1}+p_{\ell})}{p_{\ell+1}-p_{\ell}},$$
and
$$r_i-r_j = \sum_{\ell =
i}^{n-1} \frac{p_{\ell+1}+p_{\ell}}{p_{\ell+1}-p_{\ell}} -
\sum_{\ell =
  j}^{n-1}
\frac{p_{\ell+1}+p_{\ell}}{p_{\ell+1}-p_{\ell}}.$$
Therefore $i \ge j$ iff $b_i - b_j \ge 0$ and $r_i - r_j \le 0$ ---
which implies that $\left|(b_i - b_j) - (r_i - r_j)\right| = \left|b_i
- b_j\right| + \left|r_i - r_j\right|$ and thus the span $S$ of $D'_i(j)$
is equal to $S = \left|B_i - B_j\right| + \left|R_i - R_j\right|$. Furthermore,
$$\left|b_i - b_j\right| =  \sum_{\ell =
\min(i,j)}^{\max(i,j)-1}
\frac{2-(p_{\ell+1}+p_{\ell})}{p_{\ell+1}-p_{\ell}} \le
\frac{2|i-j|}{\epsilon}.$$
Analogously, $\left|r_i - r_j\right| \le
\frac{2|i-j|}{\epsilon}$. 
It follows that the span of $D'_i(j)$ can be upper bounded by
$$S = |B_i - B_j| + |R_i - R_j| \le
\frac{4|i-j|}{\epsilon M}.$$

\smallskip

Observe that $D_i(j)$ is a linear function of the number $m_b$ of blue
balls picked (and the number $m_r = m -m_b$ of red balls picked):
$$D_i(j) = m_b \cdot (B_i - B_j) + m_r \cdot (R_i - R_j).$$
Therefore, $$E[D_i(j)] = m \cdot p_i  \cdot (B_i - B_j) + m \cdot
(1-p_i) \cdot (R_i - R_j) = m \cdot \Delta_i(j),$$
where $\Delta_i(j)$ is the functional defined in
Section~\ref{sec:twosig}; recall that we proved there that
$\Delta_i(j) \ge \frac{|i-j|}M$.

Recall that with probability $1-\eta$, $\left|m_b - m \cdot p_i\right| \le
\frac{\epsilon}5 \cdot m$. If this event happens we have that, for each $j \ne i$,
$$D_i(j) \ge E[D_i(j)] - \frac{\epsilon}5 \cdot m \cdot S \ge m \cdot \left(\Delta_i(j) -
 \frac45 \cdot \frac{|i-j|}{M}\right) = m \cdot \frac{|i-j|}{5M} > 0.$$
Therefore, urn $i$ will beat each urn $j \ne i$, with probability $1-\eta$. The
proof is concluded.
\end{proof}

Observe that the previous bound is tight in a strong sense: no
algorithm that picks $o\left(\epsilon^{-2} \ln
\frac1{\eta}\right)$ balls, and produces a guess arbitrarily after
having seen all their colors, is able to guess the right urn with
probability at least $1-\eta$.


\subsection{Condorcet Voting}

In this section we present a conjecture, and we elaborate on it, with
the aim of showing that {\em Condorcet voting} is as good as Cumulative voting --- and is thus optimal. We begin by recalling the definition of the Condorcet voting system:
\begin{definition}[Condorcet Voting]
In a Condorcet election, each voter returns a
(total) ordering of the candidates. Given two candidates $i$ and $j$, we say
that $i$ beats $j$ in a run-off election if more than half the voters ranked $i$ higher than
$j$. If there exists a candidate $i$ that beats each other candidate
$j \ne i$ in a run-off election, then $i$ is the winner of the
Condorcet election.
\end{definition}
We observe that, in Condorcet voting, voters do not assign real numbers to candidates as in Cumulative voting --- they rather return a discrete object: a permutation of them.

There exist many variants of the Condorcet
election. The differences between them lie in the way of dealing with
ties (that is, when no candidate $i$ beats each
other candidate $j$ in a run-off election). Our main theorem holds for
any such variant, since our theorem will guarantee that no ties will
exist with high probability.

We start by defining a set of coefficients that will be useful for introducing a Condorcet voting scheme.

\begin{definition}\label{recb}
For $k \ge 0$ and $\ell \ge 0$, let $c_{k,\ell}$ be:
\begin{align*}
c_{k,\ell} &= (-1)^{k + \ell} \cdot \binom{k+\ell}k -\frac{(-1)^{\ell} \cdot \binom{k+\ell}k+(-1)^{k+\ell}\cdot(\ell+1)}{k+\ell+2}.
\end{align*}
We define inductively the sequence $b_{k,\ell}$, indexed by two integers; if $k < 0$ or $\ell <0$ then $b_{k,\ell} = 0$; for notational simplicity we will use:
\begin{align*}
x_{k,\ell} = \sum_{\substack{i,j\ge
      0\\i+j \le k }} \frac{b_{i,k-i-j}\cdot
    b_{j,\ell}}{(i+1)(i+j+2)}, &\quad \quad
y_{k,\ell} =\sum_{i=0}^{k}\left((-1)^{k-i}\cdot\frac{b_{i,\ell}}{i+1}\right).
\end{align*}
Then, for $k \ge 0, \ell \ge 0$, we define $b_{k,\ell} = (k+1) \cdot
\left(x_{k-1,\ell}+y_{k-1,\ell}-c_{k+1,\ell}\right)$.
\end{definition}

The induction is well-defined:
$x_{k-1,\ell}$ only depends on the $b_{i,j}$'s for which either (a) $i+j
\le k-1$, or (b) $i \le k-1$ and $j = \ell$; $y_{k-1,\ell}$ only depends on
the $b_{i,j}$'s for which $i \le k-1$ and $j = \ell$.
Therefore, the $b_{k,\ell}$'s can be computed in the following order
via the recurrence: for $n = 0, 1, 2, \ldots$ and for $k = 0,
\ldots, n$, compute $b_{k,n-k}$.
We now make a conjecture on the $b_{k,\ell}$'s:
\begin{conjecture}\label{condorcet_conjecture}
Let
$B(x,P) = \sum_{k = 0}^{\infty} \sum_{\ell=0}^{\infty} \left(b_{k,\ell} \cdot x^k \cdot P^{\ell}\right).$
Then,

\smallskip

(i) the series $B(x,P)$ converges for $0 \le x < \frac12$, $x \le P \le \frac12$, 

(ii) $B(x,P) \ge 0$ for $0 \le x < \frac12$, $x \le P \le \frac12$, and

(iii) $\int_{0}^P B(x,P) \, dx = \frac{1}{P + 1} - \sqrt{\frac{1 - 2P}{1 +2P}}$, for $0 \le P \le \frac12$.
\end{conjecture}
\begin{proof}[Comments on Conjecture~\ref{condorcet_conjecture}]
We now make some comments on the conjecture, indicating some possible approaches to settle it.
\begin{enumerate}
\item[(i)] Numerical approximations indicate
  that $B(x,P)$ converges for each $0 \le x < \frac12$, $x \le P \le \frac12$, but it diverges for $x = P = \frac12$.

\item[(ii)] For $k \ge 0$, and $0 \le \ell \le k$, let
$$a_{k,\ell} = \sum_{i=0}^{\ell}\left(\binom{k+2}i \cdot b_{k,\ell-i}\right),$$
also, let
$$B_1(x, P) = \sum_{k=0}^{\infty} \left( \frac{\sum_{\ell=0}^k
  \left(a_{k,\ell} \cdot P^{\ell}\right)}{ (P+1)^{k+2}} \cdot x^k \right).$$

By looking at the first few terms of $B_1(x,P)$'s Taylor expansion, it appears that $B(x,P) = B_1(x,P)$.
\comment{
\begin{eqnarray*}
B(P,Q) &=& \frac1{ (Q+1)^2} + \frac{2+4\,Q}{ (Q+1)^3} \cdot P +
\frac{6+8\,Q+8\,Q^2}{2 \cdot (Q+1)^4} \cdot P^2+
\frac{40+112\,Q+80\,Q^2+32\,Q^3}{3! \cdot (Q+1)^5} \cdot P^3+\\
&& \frac{200+688\,Q+1200\,Q^2+848\,Q^3+256\,Q^4}{4! \cdot (Q+1)^6} \cdot P^4 +\\
&&
\frac{2352+9712\,Q+16096\,Q^2+15744\,Q^3+8128\,Q^4+1840\,Q^5}{5! \cdot
  (Q+1)^7
}
\cdot P^5 + \\
&& \sum_{k=6}^{\infty} \left( \frac{\sum_{\ell=0}^k \left(a_{k,\ell}
  \cdot Q^{\ell}\right)}{(Q+1)^{k+2}} \cdot P^k \right).
\end{eqnarray*}}

The $B_1(x,P)$ expression, if $B_1(x,P)=B(x,P)$, could be quite useful to prove non-negativity, since the
$a_{k,\ell}$'s seem all to be non-negative --- if they are point (ii) of the conjecture directly follows.

Also, if $B_1(x,P) = B(x,P)$, then one can express
$b_{k,\ell}$ (for each $\ell \ge 0$) in terms of $b_{k,0},b_{k,1},\ldots, b_{k,k}$:
$$b_{k,\ell} = \sum_{i=0}^k \left( b_{k,i} \cdot \sum_{j=i}^{\min(k,\ell)} \left(
\frac{(-1)^{\ell-j}}{(\ell-j)!} \cdot \binom{k+2}{j-i} \cdot
\sum_{h=0}^{\ell-j}\left(\left[\ell - j \atop h\right] \cdot
(k+2)^h\right)\right)\right),$$
where $\left[n \atop k\right]$ represents the
unsigned Stirling number of the first kind with indices $n \ge k$. The
last claim can be proved using the $B_1(x,P)$ expression and the
following expression for $(Q+1)^{-t}$, $t > 0$:
$$\frac1{(Q+1)^t} = \sum_{i=0}^{\infty} \left(\frac{(-1)^i}{i!} \cdot
\sum_{j=0}^i \left( \left[i \atop j\right] \cdot t^j\right)\cdot Q^i \right).$$

\item[(iii)] Since
$$\frac1{P+1} - \sqrt{\frac{1-2P}{1+2P}} = \sum_{n=0}^{\infty}
\left(C_n \cdot P^n\right),$$
with
$$C_n = \left\{\begin{array}{cc}
1-\binom{2\left\lfloor \nicefrac n2\right\rfloor}{\left\lfloor
  \nicefrac n2\right\rfloor} & \text{if } n \text{ is
  even}\\
2\binom{2\left\lfloor \nicefrac n2\right\rfloor}{\left\lfloor
  \nicefrac n2\right\rfloor} - 1 & \text{if } n \text{ is
  odd}
\end{array}\right. = \frac{1 - 3 \cdot (-1)^n}2 \cdot
\binom{2\left\lfloor \nicefrac n2\right\rfloor}{\left\lfloor \nicefrac
  n2\right\rfloor} + (-1)^n,$$
we have that the point (iii) of the conjecture states that, for $n \ge 0$, 
$$\sum_{k = 0}^{n} \frac{b_{k,n-k}}{k+1} = C_{n+1}.$$
\end{enumerate}

Tables~\ref{bkl} and~\ref{akl} show how the $b_{k,\ell}$ and the $a_{k,\ell}$ sequences begins.

\begin{table}
\begin{center}
\begin{tabular}{|c|c|c|c|c|c|c|}
\hline
$b_{k,\ell}$ & $\ell=0$& $\ell=1$& $\ell=2$& $\ell=3$& $\ell=4$&
$\ell=5$\\
\hline
$k=0$ &$1$& $-2$& $3$& $-4$& $5$& $-6$\\\hline
$k=1$& $2$& $-2$& $0$& $4$& $-10$& $18$\\\hline
$k=2$&   $3$& $-8$& $18$& $-36$& $65$& $-108$\\\hline
$k=3$&   $\frac{20}3$ & $-\frac{44}3$ & $20$ & $-\frac{44}3$ & $-\frac{40}3$ &$80$\\\hline
$k=4$&$\frac{25}3$&$-\frac{64}3$& $53$& $-\frac{388}3$& $\frac{880}3$&$-610$\\\hline
$k=5$&$\frac{98}5$&$-\frac{844}{15}$&$\frac{582}5$&$-188$&$\frac{668}3$&$-\frac{558}5$\\\hline
\end{tabular}
\end{center}
\caption{The first few $b_{k,\ell}$'s.}
\label{bkl}
\end{table}

\begin{table}
\begin{center}
\begin{tabular}{|c|c|c|c|c|c|c|}
\hline
$a_{k,\ell}$ & $\ell=0$& $\ell=1$& $\ell=2$& $\ell=3$& $\ell=4$&
$\ell=5$\\
\hline
$k=0$ &$1$& & & & & \\\hline
$k=1$& $2$& $4$& & & & \\\hline
$k=2$&   $3$& $4$& $4$& & & \\\hline
$k=3$&   $\frac{20}3$ & $\frac{56}3$ & $\frac{40}3$ & $\frac{16}3$ &  &\\\hline
$k=4$&$\frac{25}3$&$\frac{86}3$& $50$& $\frac{106}3$& $\frac{32}3$&\\\hline
$k=5$&$\frac{98}5$&$\frac{1214}{15}$&$\frac{2012}{15}$&$\frac{656}5$&$\frac{1016}{15}$&$\frac{46}3$\\\hline
\end{tabular}
\end{center}
\caption{The first few $a_{k,\ell}$'s.}
\label{akl}
\end{table}

\comment{
\say{
The following expression could be useful:
\begin{align*}
&\int_0^P \left(B_1(x,Q) \cdot \int_0^x B_1(y,P) dy \, dx\right) - \frac1{P+1} \cdot \int_0^P B_1(x,Q) dx\\
=&\sum_{k=2}^{\infty} \left(\sum_{i=0}^{k-2} \left(\frac1{i+1}
\cdot \frac{\sum_{j=0}^i \left(a_{i,j} \cdot P^j\right)\cdot
  \sum_{j=0}^{k-i-2} \left(a_{k-i-2,j} \cdot Q^j\right)}{(P+1)^{i+2}
  (Q+1)^{k-i}} \right) \cdot \frac{P^k}k\right) -
\sum_{k=1}^{\infty}\left(\frac{\sum_{j=0}^{k-1}\left(a_{k-1,j} \cdot Q^j\right)}{(P+1)(Q+1)^{k+1}} \cdot \frac{P^k}{k}\right)
\end{align*}}}
\end{proof}

We use Conjecture~\ref{condorcet_conjecture} to show the existence of a probability distribution over permutations that induces a given set of ``marginals''.
\begin{lemma}\label{lemma:permutation}
Let $0 \le p_1 < p_2 < \cdots < p_n  \le 1$. Then, if Conjecture~\ref{condorcet_conjecture} is true, there exists a probability distribution over the
symmetric group $S_{n}$, such that, for each $1 \le i < j \le n$,
$$\Pi_{i,j} = \Pr_{\pi}\left[\pi(i) < \pi(j)\right] = \min\left(1,
\frac{1}{p_i + p_j}\right).$$
\end{lemma}
\begin{proof}
To each probability $p \in [0,1]$ we associate a random variable $X_{p}$ with values:\begin{itemize}
\item if $p \le \frac12$, $X_{p}$ is the constant random variable $p$, with
  a point mass $\gamma_p = 1$ at $p$;
\item if $p > \frac12$, $X_{p}$ has a point mass of $0 \le \gamma_p \le 1$ at $p$; $X_{p}$ has a density function
  $f_{p}(x)$, with $f_{p}(x) = \alpha_p(x)$ for $x \in
  \left[1-p,\frac12\right]$ and $f_{p}(x) = \beta_p(x)$ for $x \in
  \left(\frac12,p\right)$, with
$$\alpha_{p}(x) = (p+x)^{-2},$$
$$\beta_p(x) = B\left(x-\frac12,p-\frac12\right),$$
and
$$\gamma_p = \sqrt{\frac1p - 1}.$$
\end{itemize}

By Conjecture~\ref{condorcet_conjecture}, we have that $\beta_p(x) \ge
0$ and $\int_{\nicefrac12}^{p} \beta_p(x) \, dx =
\frac2{2p+1}-\sqrt{\frac1p-1}$; therefore the total probability mass
assigned to $X_p$ is $1$ --- $X_p$ is then a well-defined random variable.

\smallskip

The CDF associated to $\alpha_p(x)$, for $p \ge \frac12$ and $1-p \le y \le \frac12$ is:
$$\int_{1-p}^y \alpha_p(x) dx = \frac{y-1+p}{y+p}.$$

The CDF associated to $\beta_p(x)$, for $p \ge \frac12$ and $\frac12 \le y \le p$ is:
$$\int_{\frac12}^y \beta_p(x) dx =
\sum_{k=0}^{\infty}\sum_{\ell=0}^{\infty}\left(\frac{b_{k,\ell}}{k+1}\cdot
\left(y - \frac12\right)^{k+1} \cdot \left(p - \frac12\right)^{\ell}\right)$$

\smallskip

Observe that if $p \ne q$, then there's zero probability that $X_p =
X_q$ --- since, for each $p$, $X_p$ has a positive point mass only at
$p$. Then, we pick a permutation $\pi$ of $\{1,2,\ldots, n\}$ by letting $\pi(i) < \pi(j)$
iff $X_{p_i} < X_{p_j}$.

\comment{
\say{
Regarding $\beta_p(x)$, I believe that:

$$\int_{\frac12}^p \beta_p(x) dx = 1 - \left(\gamma_p + \int_{1-p}^{\frac12} \alpha_p(x) dx\right)= \frac{2}{2p+1} - \sqrt{\frac1p - 1}.$$

$$\beta_p\left(\frac12\right) = \alpha_p\left(\frac12\right) = \frac1{\left(p+\frac12\right)^2}.$$
}}

We verify that the marginals $\Pi_{i,j}$ of our distribution satisfy
the requirement in the claim:
\begin{itemize}
\item if $0 \le p < q \le 1 - p$, then:
$$\Pr\left[X_p \le X_q\right] = 1$$
\item if $0 \le p < \frac12$ and $1-p \le q \le 1$, then:
$$\Pr\left[X_p \le X_q\right] = \Pr\left[X_q \ge p\right] = 1 - \int_{1-q}^{p} \left(q+x\right)^{-2} dx  = \frac{1}{p+q}.$$
\item if $\frac12 < p < q \le 1$. Let $P_1 = \Pr\left[X_p
  \le X_q \le \frac12\right]$, $P_2 = \Pr\left[X_p \le \frac12 \le
  X_q\right]$, $P_3 = \Pr\left[\frac12 \le X_p \le p \le X_q\right]$, and $P_4 = \Pr\left[\frac12 \le X_p \le
  X_q < p\right]$. Then,
$$\Pr\left[X_p \le X_q\right] = P_1 + P_2 + P_3 + P_4.$$
\end{itemize}
We now show that $P_1+P_2+P_3+P_4 = \frac1{p+q}$, completing the
proof. We start by computing $P_1$ and $P_2$:
\begin{align*}
P_1 &= \int_{1-p}^{\frac12} \left((q+x)^{-2} \cdot \int_{1-p}^{x}
(p+y)^{-2} dy \right)dx = \int_{1-p}^{\frac12} \frac{x-(1-p)}{(q+x)^2(p+x)}
dx \\
& =\frac{2p-1}{(q-p)(2q+1)}-
\frac{\ln\frac{(2p+1)(1-p+q)}{2q+1}}{(q-p)^2}.
\end{align*}
$$P_2 = \int_{1-p}^{\frac12} (p+x)^{-2} dx \cdot \left(1-
\int_{1-q}^{\frac12} (q+x)^{-2} dx\right) = \frac{2p-1}{2p+1} \cdot
\left(1-\frac{2q-1}{2q+1}\right) = \frac{2p-1}{2p+1} \cdot
\frac{2}{2q+1}.$$

As for $P_3$ and $P_4$, we have:
$$P_3 = \left(\int_{\frac12}^p \beta_p(x) dx + \sqrt{\frac1p - 1}\right)\cdot \left(\int_{p}^q \beta_q(x) dx + \sqrt{\frac1q-1}\right) = \frac2{2p+1} \cdot \left(\frac2{2q+1} - \int_{\frac12}^p \beta_q(x) dx\right).$$
$$P_4 = \int_{\frac12}^p \left(\beta_q(x) \cdot \int_{\frac12}^x \beta_p(y) dy\right)dx$$

We show that for our definition of $B(x,P) = \beta_{P+\nicefrac12}\left(x+\nicefrac12\right)$, the equation
$\Pr\left[X_p \le X_q\right] = P_1 + P_2 + P_3 + P_4$ is equal to
$P_1+P_2+P_3+P_4 = \frac1{p+q}$, for each $\frac12 < p < q \le 1$.

The latter is true iff $P_3 + P_4 = \frac1{p+q} - P_1 - P_2$;
expanding the terms, we get the equation holds iff:
\begin{eqnarray*}
\frac2{2p+1} \cdot \left(\frac2{2q+1}-\int_{\frac12}^p \beta_q(x) dx \right) + \int_{\frac12}^p \left(\beta_q(x) \cdot \int_{\frac12}^x \beta_p(y) dy\right)dx =\\
 \frac1{p+q} - \frac{2p-1}{(q-p)(2q+1)}+ \frac{\ln\frac{(q-p+1)(2p+1)}{2q+1}}{(q-p)^2} 
 -  \frac{2p-1}{2p+1} \cdot
\frac{2}{2q+1},
\end{eqnarray*}
or, equivalently,
\begin{eqnarray*}
\int_{\frac12}^p \left(\beta_q(x) \cdot \int_{\frac12}^x \beta_p(y) dy\right)dx -\frac2{2p+1} \cdot \int_{\frac12}^p \beta_q(x) dx  =\\
 \frac1{p+q} + \frac{\ln\frac{(q-p+1)(2p+1)}{2q+1}}{(q-p)^2} -
\frac{2q-1}{(q-p)(2q+1)}.
\end{eqnarray*}
\comment{
Together with the constraints, we get a differential system (with unknowns $\beta_p(x)$, for $\frac12 < p\le 1$):
$$\left\{\begin{array}{rl}
\int_{\nicefrac12}^p \left(\beta_q(x) \cdot \int_{\nicefrac12}^x \beta_p(y) dy\right)dx -\frac2{2p+1} \cdot \int_{\nicefrac12}^p \beta_q(x) dx  =\\
= \frac1{p+q} + \frac{\ln\frac{(q-p+1)(2p+1)}{2q+1}}{(q-p)^2} -
\frac{2q-1}{(q-p)(2q+1)} & \text{for each } \nicefrac12 < p < q \le 1\\
\int_{\nicefrac12}^p \beta_p(x) dx = \frac{2}{2p+1}-\sqrt{\frac1p-1} & \text{for each } \nicefrac12 <p\le1\\
\beta_p(x) \ge 0 & \text{for each } \nicefrac12 < x < p \le 1
\end{array}\right.$$

We can also add $\beta_p\left(\frac12\right) = \left(p + \frac12\right)^{-2}$, for each $\frac12 < p \le 1$, to the system. But this last constraint is possibly implied by the ones already in the system (and in any case, it is not important that this property holds --- only, the script's output indicates that it does).

\medskip
}
For notational convenience, we let $P = p - \frac12$ and $Q = q -
\frac12$. Let $B(P,Q) = \beta_q(p) = \beta_{Q + \frac12}(P +
\frac12)$. 
\comment{
We aim to find the coefficients $b_{k,\ell}$. Consider, for each $\frac12 < p < q \le 1$, the first equation of the previous differential system:
\begin{eqnarray*}
(q-p)^2 \cdot \left(\int_{\frac12}^p \left(\beta_q(x) \cdot \int_{\frac12}^x \beta_p(y) dy\right)dx -\frac2{2p+1} \cdot \int_{\frac12}^p \beta_q(x) dx \right) =\\
\frac{(q-p)^2}{p+q} + \ln\frac{(q-p+1)(2p+1)}{2q+1} -
\frac{(q-p)(2q-1)}{2q+1}.
\end{eqnarray*}
By substituting $p$ with $P + \frac12$ and $q$ with $Q+\frac12$, we
get}
We thus need to prove:
\begin{align}
\int_{0}^P \left(B(x,Q) \cdot \int_{0}^x B(y,P) dy\right)dx
-\frac1{P+1} \cdot \int_{0}^P B(x,Q) dx = \nonumber\\
 \frac{1}{P+Q+1} + \frac{\ln\frac{(Q-P+1)(P+1)}{Q+1}}{(Q-P)^2} -
\frac{Q}{(Q+1)(Q-P)}.\label{diffeq}
\end{align}
Since we will expand the right-hand side in a Taylor series around $(Q,P)=(0,0)$, we observe that the right-hand side has a removable singularity at $P =
Q$. Indeed, letting $P=Q-\epsilon$,
\begin{align*}
\lim_{\epsilon \rightarrow 0} \left(\frac{\ln\frac{(1+\epsilon)(Q+1-\epsilon)}{Q+1}}{\epsilon^2} -
\frac{Q}{(Q+1)\epsilon} \right)&= \lim_{\epsilon \rightarrow 0}\left(
\frac{\ln\left((1+\epsilon)\left(1- \frac{\epsilon}{Q+1}\right)\right)}{\epsilon^2} -
\frac{Q}{(Q+1)\epsilon}\right) \\
&= \lim_{\epsilon \rightarrow 0} \left(\frac{
  \epsilon - \frac{\epsilon}{Q+1} - \frac{\epsilon^2}{Q+1} -
  \frac{Q^2\epsilon^2}{2(Q+1)^2} + O(\epsilon^3)}{\epsilon^2} - \frac{Q}{(Q+1)\epsilon}\right)\\
&= \lim_{\epsilon \rightarrow 0} \frac{-
   \frac{\epsilon^2}{Q+1} -
  \frac{Q^2\epsilon^2}{2(Q+1)^2} + O(\epsilon^3)}{\epsilon^2}\\
& = -\frac1{Q+1}- \frac{Q^2}{2(Q+1)^2},
\end{align*}
since the limit from the left and the right coincide.
Therefore, if $P = Q$, the right-hand side of Equation~\ref{diffeq} is 
$$\frac1{2Q+1}-\frac1{Q+1}- \frac{Q^2}{2(Q+1)^2}.$$

\smallskip

Recalling that
$$B(P,Q) = \sum_{\ell = 0 }^{\infty} \sum_{k=0}^{\infty}
\left(b_{k,\ell} \cdot P^k \cdot Q^{\ell}\right),$$
we get that the following holds:
\begin{align*}
\int_0^P\left(B(x,Q)\cdot\int_0^xB(y,P)dy\right)dx &=
\sum_{\ell=0}^{\infty}\sum_{k=2}^{\infty}\left[\left(\sum_{\substack{i,j\ge
      0\\i+j \le k-2 }} \frac{b_{i,k-2-i-j}\cdot
    b_{j,\ell}}{(i+1)(i+j+2)}\right) \cdot P^k \cdot Q^{\ell}\right]
\\
&=
\sum_{\ell=0}^{\infty}\sum_{k=2}^{\infty}\left(x_{k-2,\ell} \cdot P^k
\cdot Q^{\ell}\right).
\end{align*}
\comment{We also have:
\begin{align*}
(Q-P)^2 \cdot \int_0^P\left(B(x,Q)\cdot\int_0^xB(y,P)dy\right)dx &=
  \left(Q^2 - 2PQ + P^2\right) \cdot
\sum_{\ell=0}^{\infty}\sum_{k=0}^{\infty}\left(x_{k-2,\ell} \cdot P^k
\cdot Q^{\ell}\right)\\
&=  \sum_{\ell=0}^{\infty}\sum_{k=0}^{\infty}\left( \left(x_{k-2,\ell-2}-2x_{k-3,\ell-1}+x_{k-4,\ell}\right) \cdot P^k
\cdot Q^{\ell}\right).
\end{align*}}
Also,
$$\int_0^P B(x,Q) dx = \sum_{\ell=0}^{\infty}\sum_{k=1}^{\infty}\left(\frac{b_{k-1,\ell}}k \cdot P^k \cdot Q^{\ell}\right),$$
and, since $\frac1{P+1} = \sum_{k=0}^{\infty} \left( (-1)^k \cdot P^k\right)$, 
$$ - \frac1{P+1} \cdot \int_0^P B(x,Q) dx =
-\sum_{\ell=0}^{\infty}\sum_{k=1}^{\infty}\left[\sum_{i=0}^{k-1}\left((-1)^{k-1-i}\cdot\frac{b_{i,\ell}}{i+1}\right)
  \cdot P^k \cdot Q^{\ell}\right] =
-\sum_{\ell=0}^{\infty}\sum_{k=1}^{\infty}\left(y_{k-1,\ell} \cdot P^k
\cdot Q^{\ell}\right) .$$
\comment{
Then,
\begin{align*}
 - \frac{(Q-P)^2}{P+1} \cdot \int_0^P B(x,Q) dx & =
 \left(Q^2-2QP+P^2\right) \cdot
 \sum_{\ell=0}^{\infty}\sum_{k=1}^{\infty}\left(h_{k,\ell}
   \cdot P^k \cdot Q^{\ell}\right)\\
&=
-\sum_{\ell=0}^{\infty}\sum_{k=1}^{\infty}\left(\left(y_{k-1,\ell-2}-2y_{k-2,\ell-1}+y_{k-3,\ell}\right) \cdot P^k
\cdot Q^{\ell}\right)
\end{align*}
}
Furthermore, we have
\begin{eqnarray*}
\frac{1}{P+Q+1} &=& \sum_{k=0}^{\infty} \sum_{\ell=0}^{\infty} c_1(k,\ell) \cdot P^k \cdot Q^{\ell}\\
\frac{\ln \frac{(Q-P+1)(P+1)}{Q+1}}{(Q-P)^2} - \frac{Q}{(Q+1)(Q-P)} & = &\sum_{k=0}^{\infty} \sum_{\ell=0}^{\infty} c_2(k,\ell) \cdot P^k \cdot Q^{\ell}\\
\end{eqnarray*}
with
$$c_1(k,\ell) = (-1)^{k + \ell} \cdot \binom{k+\ell}k,$$
$$c_2(k,\ell) = -\frac{(-1)^{\ell} \cdot \binom{k+\ell}k+(-1)^{k+\ell}\cdot(\ell+1)}{k+\ell+2}.$$
observe that $c_{k,\ell} = c_1(k,\ell) + c_2(k,\ell)$, where $c_{k,\ell}$ is as in  Definition~\ref{recb}.

\comment{
Then,
\begin{eqnarray*}
c(k,\ell) &=& \frac{(-1)^{k+\ell}\binom{k+\ell}k\cdot\left((k+\ell)\cdot\left(1-(-1)^k\right)+(-1)^k\right)+4(-1)^{k+\ell+1}\binom{k+\ell-2}{k-1}(k+\ell)+[l=0]\cdot (-1)^k}{k+\ell} +\\
&& [k \le 2] \cdot (-1)^{\ell+1} - [k=1 \wedge \ell =0]\\
c(k,\ell) &=& \binom{k+\ell}k (-1)^{\ell}\cdot\left((-1)^k-1+\frac{1}{k+\ell}\right)
-4(-1)^{k+\ell}\binom{k+\ell-2}{k-1}+\\
&&[l=0]\cdot\frac{ (-1)^k}{k} + [k \le 2] \cdot (-1)^{\ell+1} - [k=1
  \wedge \ell =0].
\end{eqnarray*}
Observe that $c(0,\ell) = 0$ for each $\ell \ge 0$, and that $c(1,0) =
c(1,1) = c(2,0) = 0$.

 If we choose $k \ge 1, \ell \ge 0$ in such a way
that $(k,\ell) \not\in \left\{(1,0),(1,1),(2,0)\right\}$, we obtain
$$c(k,\ell) = \binom{k+\ell}k (-1)^{\ell}\left((-1)^k-1+\frac{1}{k+\ell}\right)
-4  (-1)^{k+\ell}\binom{k+\ell-2}{k-1}+ \left\{\begin{array}{cr}
(-1)^{\ell+1} & k \le 2\\
\frac{(-1)^k}k & k > 2 \wedge \ell = 0\\
0 & \text{otherwise}
\end{array}\right.$$

\medskip
}
We then have that $P_1+P_2+P_3+P_4 = \frac1{p+q}$ (particularly, Equation~\ref{diffeq}) is satisfied iff
$$\sum_{\ell=0}^{\infty} \sum_{k=0}^{\infty} \left(\left(x_{k-2,\ell}-y_{k-1,\ell}\right) \cdot P^k \cdot Q^{\ell}\right) = \sum_{\ell=0}^{\infty} \sum_{k=0}^{\infty} \left(c_{k,\ell} \cdot P^k \cdot Q^{\ell}\right).$$
The equality holds iff for each $k, \ell \ge 0$, we have 
$$c_{k,\ell} =x_{k-2,\ell}-y_{k-1,\ell}.$$
Observe that if $k = 0$, then the equation is satisfied, since
$c_{0,\ell}=x_{-2,\ell}=y_{-1,\ell}=0$. If $k \ge 1$, then
$y_{k-1,\ell} = \frac{b_{k-1,\ell}}k - y_{k-2,\ell}$. Since we defined the
$b_{k,\ell}$'s to be, for
for each $k, \ell \ge 0$, 
$$b_{k,\ell} = (k+1) \cdot
\left(x_{k-1,\ell}+y_{k-1,\ell}-c_{k+1,\ell}\right),$$
we have that Equation~\ref{diffeq} is satisfied and $P_1+P_2+P_3+P_4 = \frac1{p+q}$.
\end{proof}

\comment{
\say{This has to be checked. One of these equations is wrong.

Another set of equations that seem to hold for the $b_{k,\ell}$'s: for each $k \ge 0$,
$$\sum_{\ell=0}^k \left( \frac1{\ell+1} \cdot \binom{k}{\ell} \cdot b_{k,\ell}\right) = 1$$
$$\sum_{\ell=0}^k \left( \binom{k}{\ell} \cdot b_{k,\ell}\right) = k^2 + (-1)^k$$
$$\sum_{\ell=0}^k \left( \ell \cdot \binom{k}{\ell} \cdot b_{k,\ell}\right) = \frac{k^4 + 3k^2 - 4k + 2}2 +(-1)^k\cdot \left(k^2 + 3k - 1\right)$$
$$\sum_{\ell=0}^k \left( \ell^2 \cdot \binom{k}{\ell} \cdot b_{k,\ell}\right) = \frac{k^6 + 10k^4 - 36k^3 + 79k^2 -60k+18}6 +(-1)^k\cdot \frac{k^4+6k^3+3k^2-12k+10}{2}$$
}
}

Using the previous distribution over permutations we can prove the main theorem of the section.
\begin{theorem}\label{condorcet_theorem}
Let $0 \le p_1 < p_2 < \cdots < p_n  \le 1$ be the blue-probabilities
of a set of bichromatic urns. Let $\epsilon = \min_{1 \le i \le n-1}
(p_{i+1} - p_i)$. If Conjecture~\ref{condorcet_conjecture} is true, there exists a symmetric voting scheme for
the Condorcet election that
guarantees that the unknown urn wins with probability $1-\eta$ with
$O\left(\frac{\ln \eta^{-1}}{\epsilon^2}\right)$ voters.
\end{theorem}
\begin{proof}
Lemma~\ref{lemma:permutation} guarantees the existence of a
probability distribution $P$ over the set of permutations of
$\{1,2,\ldots, n\}$ such that, for each $1 \le i < j \le n$, $\Pr_{\pi
  \sim P}[\pi(i)< \pi(j)] = \min\left(1,\frac1{p_i+p_j}\right)$. If
$\pi(j) > \pi(i)$ we say that $j$ beats $i$ in $\pi$.

We also let $q_i = 1 - p_{i}$; therefore $0\le q_n < q_{n-1} < \cdots
< q_1 \le 1$, and $\min_{1 \le i \le n-1} (q_{i+1} - q_i) = \epsilon$.
 Lemma~\ref{lemma:permutation} again guarantees the existence of a
 probability distribution $Q$ over the set of permutations of
 $\{1,2,\ldots, n\}$ such that for $n \ge i > j \ge 1$, we have
 $\Pr_{\pi \sim Q}[\pi(i) < \pi(j)] =
 \min\left(1,\frac{1}{q_i+q_j}\right)$. 

Each voter will apply the following algorithm: if she draws blue, she
sample a permutation according to $P$, otherwise she samples a
permutation according to $Q$.

Now, suppose the $i$-th urn is the unknown urn. Let $j \ne i$ be the
index of any other urn. If $j > i$,
\begin{align*}
\Pr[\text{the unknown urn } i \text{ beats another urn } j] &= p_i \Pr_{\pi \sim P}[\pi(i)
  > \pi(j)] + (1-p_i) \Pr_{\pi \sim Q}[\pi(i) > \pi(j)]\\
&= p_i \max\left(0, 1-\frac{1}{p_i+p_j}\right) + q_i
\min\left(1,\frac{1}{q_i+q_j}\right),
\end{align*}
if $p_i + p_j \le 1$ (and therefore $q_i + q_j \ge 1$) the latter simplifies to
$\frac{q_i}{q_i+q_j}$; otherwise $p_i + p_j > 1, q_i+q_j <1$, and the
expression simplifies to $p_i\left(1-\frac{1}{p_i+p_j}\right)+q_i = p_i -
\frac{p_i}{p_i+p_j} + 1 -p_i =1-\frac{p_i}{p_i+p_j} =
\frac{p_j}{p_i+p_j}$.

Therefore, if $j > i$, we have
$$\Pr[\text{the unknown urn } i \text{ beats another urn } j] \in
\left\{\frac{q_i}{q_i+q_j}, \frac{p_j}{p_i+p_j}\right\}.$$

\smallskip

If, otherwise, $j < i$, we have
\begin{align*}
\Pr[\text{the unknown urn } i \text{ beats another urn } j] &= p_i \Pr_{\pi \sim P}[\pi(i)
  > \pi(j)] + (1-p_i) \Pr_{\pi \sim Q}[\pi(i) > \pi(j)]\\
&= p_i \min\left(1, \frac{1}{p_i+p_j}\right) + q_i
\max\left(0,1-\frac{1}{q_i+q_j}\right),
\end{align*}
if $q_i + q_j \le 1$ (and therefore $p_i + p_j \ge 1$) the latter simplifies to
$\frac{p_i}{p_i+p_j}$; otherwise $q_i + q_j > 1, p_i+p_j <1$, and the
expression simplifies to $p_i + q_i \left(1-\frac{1}{q_i+q_j}\right) =
1-q_i+q_i -\frac{q_i}{q_i+q_j} = \frac{q_j}{q_i+q_j}$.

\smallskip

Therefore, in any case, $\Pr[\text{the unknown urn } i \text{ beats
    another urn } j] \in
\left\{\frac{\max(q_i,q_j)}{q_i+q_j}, \frac{\max(p_i,p_j)}{p_i+p_j}\right\}$.
We lower-bound the latter two fractions:
\begin{align*}
\frac{\max(p_i,p_j)}{p_i+p_j} & = \frac{\max(p_i,p_j)}{2\max(p_i,p_j) - \left|p_i-p_j\right|} =
\frac{\max(p_i,p_j) - \frac12 \left|p_i - p_j\right|}{2\max(p_i,p_j) - \left|p_i-p_j\right|}  + \frac12 \cdot \frac{\left|p_i - p_j\right|}{2\max(p_i,p_j)-\left|p_i-p_j\right|}\\
&= \frac12 \cdot \left(1 + \frac{\left|p_i - p_j\right|}{p_i + p_j}\right) \ge \frac12 + \frac{\left|p_i - p_j\right|}4,
\end{align*}
and, analogously,
$$\frac{\max(q_i,q_j)}{q_i+q_j} 
\ge \frac12 + \frac{\left|q_i - q_j\right|}4.$$

Since, $\left|q_i - q_j\right| = \left|p_i - p_j\right|$, we have
$$\Pr[\text{the unknown urn } i \text{ beats another urn } j] \ge \frac12 +
\frac{\left|p_i - p_j\right|}4 \ge \frac12 + \frac{\left|i - j\right|}4 \cdot \epsilon.$$

\medskip

Now, given two urns $i, j$, let $X_i(j)$ be the random variable
counting the number of votes in which $i > j$, with $m$
voters. Observe that if $X_i(j) > \frac m2$, then urn $i$ beats urn
$j$. Also,
$$\frac m2 \le m \cdot \left(\frac12 + \frac{\left|i - j\right| \cdot
  \epsilon}4\right) \le E[X_i(j)] \le m,$$
 and
\begin{align*}
\Pr\left[\left|X_i(j) - E[X_i(j)]\right| \ge \frac{\left|i -
  j\right| \cdot \epsilon }5 \cdot m\right] & \le
\Pr\left[\left|X_i(j) - E[X_i(j)]\right| \ge \frac{\left|i -
  j\right| \cdot \epsilon }5 \cdot E[X_i(j)]\right] \\
& \le \exp\left(-\frac{\left|i-j\right|^2 \cdot \epsilon^2}{75} \cdot
E[X_i(j)]\right)\\
& \le \exp\left(-\frac{\left|i-j\right|^2 \cdot \epsilon^2}{150} \cdot
m\right).
\end{align*}

By choosing $m = \left\lceil 150 \epsilon^{-2} \ln
\frac 3{\eta}\right\rceil$, we obtain:
$$\Pr\left[\left|X_i(j) - E[X_i(j)]\right| \ge \frac{\left|i -
  j\right| \cdot \epsilon }5 \cdot m\right] \le
\exp\left(-\left|i-j\right|^2 \ln \frac3{\eta}\right) = \left(\frac{\eta}3\right)^{\left|i -
  j\right|^2} \le \left(\frac{\eta}3\right)^{\left|i - j\right|}.$$

Observe that if $\left|X_i(j) -E[X_i(j)]\right| < \frac{\left|i -
  j\right| \cdot \epsilon }5 \cdot m$, then --- by $E[X_i(j)]  \ge m
\cdot \left(\frac12 +
\frac{\left|i - j\right| \epsilon}4\right)$ --- we get $X_i(j) \ge m \cdot
\left(\frac12 + \frac{\left|i - j\right|\epsilon}{20}\right) > \frac
m2$, which implies that urn $i$ beats urn $j$.

\medskip

Applying the Union Bound over all the urns $j \ne i$, we obtain
$$\Pr[\text{urn } i \text{ does not win the election}] \le 2 \cdot
\sum_{k=1}^{\infty} \left(\frac{\eta}3\right)^{k} = 2 \cdot \frac{\nicefrac{\eta}3}{1-\nicefrac{\eta}3} \le \eta.$$
\end{proof}

\xhdr{Acknowledgments}
We thank Larry Blume,
David Easley, and Bobby Kleinberg for valuable discussions.

\bibliographystyle{plain}
\bibliography{n,urns}

\end{document}